\RequirePackage{amsmath}
\documentclass[runningheads]{llncs}
\usepackage[utf8]{inputenc}
\usepackage[T1]{fontenc}
\usepackage{lmodern}
\usepackage{microtype}
\usepackage{graphicx}
\usepackage{fontawesome}
\usepackage{enumitem}
	\setlist[enumerate]{label=(\arabic*)}

\usepackage{csquotes}
\usepackage{amssymb}
\usepackage{mathtools}
\usepackage[%
]{todonotes}
\usepackage{tikz}
	\usetikzlibrary{positioning,arrows,arrows.meta,matrix,shapes.geometric,calc}
\usepackage[%
	colorlinks,
	allcolors=blue,
	breaklinks
]{hyperref}

\spnewtheorem{fact}[theorem]{Fact}{\bfseries}{\itshape}
\spnewtheorem{construction}{Construction}{\bfseries}{}

\DeclareMathOperator{\Sh}{Shift}
\DeclareMathOperator{\Le}{L}
\newcommand{\M}{\mathcal{M}}
\newcommand{\CC}{\mathcal{C}}
\newcommand{\Eprime}{\mathcal{E}^{\prime}}
\newcommand{\Kprime}{K^{\prime}}
\newcommand{\lprime}{l^{\prime}}
\newcommand{\llprime}{l^{\prime\prime}}
\newcommand{\rprime}{r^{\prime}}

\newcommand{\kprime}{k^{\prime}}
\newcommand{\kprimeprime}{k^{\prime\prime}}
\newcommand{\pprime}{p^{\prime}}
\newcommand{\Gprime}{G^{\prime}}

\newcommand{\Kcap}{K_{\cap}}
\newcommand{\Lcap}{L_{\cap}}
\newcommand{\Rcap}{R_{\cap}}
\newcommand{\Lcup}{L_{\cup}}
\newcommand{\scrule}[3]{#1^{-1} \ltimes_{#3} #2}
\newcommand{\conrule}{\rho_1 *_E \rho_2}
\newcommand{\genconrule}{\rho_1 *_{E,k} \rho_2}
\newcommand{\cnode}{\tikz[baseline=-0.6ex]{\filldraw circle (2pt);}}

\newif\iflong
\longtrue

\newcommand{\LongShort}[2]{%
	\iflong
		#1
	\else
		#2
	\fi
}

\begin{document}

\title{A Generalized Concurrent Rule Construction for Double-Pushout Rewriting}

\author{Jens Kosiol\,\textsuperscript{\faEnvelopeO}\orcidID{0000-0003-4733-2777} \and Gabriele Taentzer\orcidID{0000-0002-3975-5238}}

\institute{%
  Philipps-Universität Marburg, Marburg, Germany\\
	\email{\{kosiolje,taentzer\}@mathematik.uni-marburg.de}
}

\authorrunning{J. Kosiol and G. Taentzer}

\maketitle

\begin{abstract}
	Double-pushout rewriting is an established categorical approach to the rule-based transformation of graphs and graph-like objects. 
	One of its standard results is the construction of concurrent rules and the Concurrency Theorem pertaining to it: 
	The sequential application of two rules can equivalently be replaced by the application of a concurrent rule and vice versa. 
	We extend and generalize this result by introducing \emph{generalized concurrent rules} (GCRs). 
	Their distinguishing property is that they allow identifying and preserving elements that are deleted by their first underlying rule and created by the second one. 
	We position this new kind of composition of rules among the existing ones and obtain a Generalized Concurrency Theorem for it. 
	We conduct our work in the same generic framework in which the Concurrency Theorem has been presented, namely double-pushout rewriting in $\M$-adhesive categories via rules equipped with application conditions. 
	
\end{abstract}

\section{Introduction}
\label{sec:introduction}
The composition of transformation rules has long been a topic of interest for (theoretical) research in graph transformation. 
Classical kinds of rule composition are the ones of \emph{parallel} and \emph{concurrent}~\cite{ER80} as well as of \emph{amalgamated rules}~\cite{BFH85}. 
Considering the \emph{double-pushout approach} to graph transformation, these rule constructions have been lifted from ordinary graphs to the general framework of $\M$-adhesive categories and from plain rules to such with application conditions~\cite{EEPT06,EHL10,GHE14,EGHLO14}. 
These central forms of rule composition have also been developed for other variants of transformation, like \emph{single-} or \emph{sesqui-pushout rewriting}~\cite{Loewe93,Loewe15,Behr19}. 

In this work, we are concerned with simultaneously generalizing two variants of sequential rule composition in the context of double-pushout rewriting. 
We develop \emph{generalized concurrent rules} (GCRs), which comprise concurrent as well as so-called \emph{short-cut rules}~\cite{FKST18}. 
The concurrent rule construction, on the one hand, is optimized concerning {\em transient} model elements: 
An element that is created by the first rule and deleted by the second does not occur in a concurrent rule.  
A model element that is deleted by the first rule, however, cannot be reused in the second one. 
A short-cut rule, on the other hand, takes a rule that only deletes elements and a monotonic rule (i.e., a rule that only creates elements) and combines them into a single rule, where elements that are deleted and recreated may be preserved throughout the process.
GCRs fuse both effects, the omission of transient elements and the reuse of elements, into a single construction. 

The reuse of elements that is enabled by short-cut rules has two distinct advantages. 
First, information can be preserved. 
In addition, a rule that reuses model elements instead of deleting and recreating them is often applicable more frequently since necessary context does not get lost: 
Considering the double-pushout approach to graph transformation, a rule with higher reuse satisfies the dangling edge condition more often in general. 
These properties allowed us to employ short-cut rules to improve model synchronization processes~\cite{FKST19,FKST20,FKMST20}. 
Our construction of GCRs provides the possibility of reusing elements when sequentially composing arbitrary rules. 
Hence, it generalizes the restricted setting in which we defined short-cut rules. 
Thereby, we envision new possibilities for application, for example, the automated construction of complex (language-preserving) editing operations from simpler ones (which are not monotonic in general). 
This work, however, is confined to developing the formal basis. 
We present our new theory in the general and abstract framework of double-pushout rewriting in $\M$-adhesive categories~\cite{LS05,EEPT06}. 
We restrict ourselves to the case of $\M$-matching of rules, though. 
While results similar to the ones we present here also hold in the general setting, their presentation and proof are much more technical. 

In Sect.~\ref{sec:motivational-example}, we introduce our running example and motivate the construction of GCRs by contrasting it to the one of concurrent rules. 
Section~\ref{sec:preliminaries} recalls preliminaries. 
In Sect.~\ref{sec:construction-gcr}, we develop the construction of GCRs. 
We characterize under which conditions the GCR construction results in a rule and prove that it generalizes indeed both, the concurrent as well as the short-cut rule constructions. 
Section~\ref{sec:generalized-concurrency-theorem} contains our main result: 
The Generalized Concurrency Theorem 
states that subsequent rule applications can be synthesized into the application of a GCR. 
It also characterizes the conditions under which the application of a GCR can be equivalently split up into the subsequent application of its two underlying rules. 
Finally, we consider related work in Sect.~\ref{sec:related-work} and conclude in Sect.~\ref{sec:conclusion}. 
\LongShort{An appendix contains additional preliminaries and all proofs.}{A long version of this paper contains additional preliminaries and all proofs~\cite{KT20}.}

\section{Running Example}
\label{sec:motivational-example}
In this section, we provide a short practical motivation for our new rule construction. 
It is situated in the context of model editing, more precisely class refactoring~\cite{Fow99}. 
Refactoring is a technique to  improve the design of a software system without changing its behavior. 
Transformation rules can be used to specify suitable refactorings of class models. 
For the sake of simplicity, we focus on the class structure here, where classes are just blobs. 
Two kinds of class relations are specified using typed edges, namely class references and generalizations; they are typed with \textsf{ref} and \textsf{gen}, respectively.
All rules are depicted in an integrated fashion, i.e., as a single graph where annotations determine the roles of the elements. 
Black elements (without further annotations) need to exist to match a rule and are not changed by its application. 
Elements in red (additionally annotated with $--$) need to exist and are deleted upon application; green elements (annotated with $++$) get newly created. 

The refactoring rules for our example are depicted in the first line of Fig.~\ref{fig:rules-class-diagram}.
The rule \emph{removeMiddleMan} removes a \textsf{Class} that merely delegates the work to the real \textsf{Class} and directs the reference immediately to this, instead. 
The rule \emph{extractSubclass} creates a new \textsf{Class} that is generalized by an already existing one; to not introduce unnecessary abstraction, the rule also redirects an existing reference to the newly introduced subclass.

\begin{figure}[t]
	\centering
	\includegraphics[width=\linewidth]{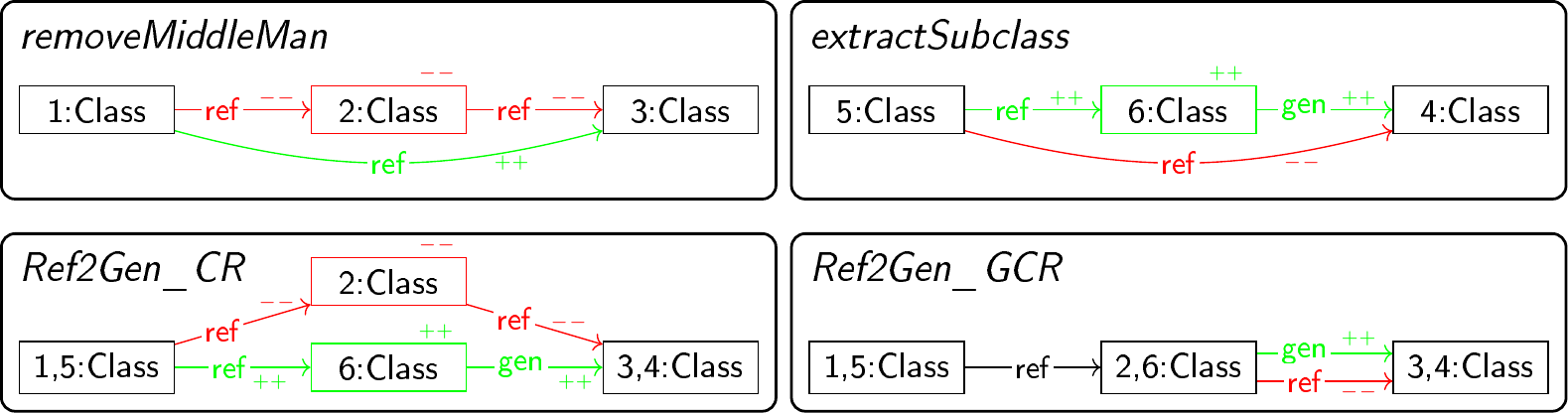}
	\caption{Two refactoring rules for class diagrams (first line) and sequentially composed rules derived from them (second line).}
	\label{fig:rules-class-diagram}
\end{figure}

Sequentially combining these two refactorings results in further ones. 
For example, this allows us to replace the second reference of a chain of two references with a generalization. 
The according concurrent rule is depicted as \emph{Ref2Gen\_CR} in Fig.~\ref{fig:rules-class-diagram}. 
It arises with \emph{removeMiddleMan} as its first underlying rule and \emph{extractSubclass} as its second, where \textsf{Classes 1} and \textsf{5} and \textsf{3} and \textsf{4} are identified, respectively. 
The new \textsf{reference}-edge created by \emph{removeMiddleMan} is deleted by \emph{extractSubclass} and, thus, becomes transient. 
But the \textsf{Class 2} that originally delegated the reference is deleted and cannot be reused. 
Instead, \textsf{Class 6} has to be newly created to put \textsf{Classes 1} and \textsf{3} into this new context. 

In many situations, however, it would be preferable to just reuse \textsf{Class 2} and only replace the reference with a generalization. 
In this way, information (such as references, values of possible attributes, and layout information) is preserved. 
And, maybe even more importantly, when adopting the double-pushout approach, such a rule is typically more often  applicable, namely also when \textsf{Class 2} has adjacent edges. 
In our construction of generalized concurrent rules, we may identify elements deleted by the first rule and recreated by the second and decide to preserve them. 
In this example, our new construction allows one to also construct \emph{Ref2Gen\_GCR} (Fig.~\ref{fig:rules-class-diagram}) from \emph{removeMiddleMan} and \emph{extractSubclass}.
In contrast to \emph{Ref2Gen\_CR}, it identifies \textsf{Class 2} deleted by \emph{removeMiddleMan} and \textsf{Class 6} created by \emph{extractSubclass} and the respective incoming references and preserves them.  This rule specifies a suitable variant of the refactoring \emph{Remove Middle Man}, where the middle man is turned into a subclass of the real class.

\section{Preliminaries}
\label{sec:preliminaries}
In this section, we recall the main preliminaries for our work. 
We introduce $\M$-adhesive categories, double-pushout rewriting, initial pushouts, $\M$-effective unions, and concurrent rules. 

Adhesive categories can be understood as categories where pushouts along monomorphisms behave like pushouts along injective functions in $\mathbf{Set}$. 
They have been introduced by Lack and Soboci\'nski~\cite{LS05} and offer a unifying formal framework for double-pushout rewriting. 
Later, more general variants, which cover practically relevant examples that are not adhesive, have been suggested~\cite{EEPT06,EGHLO14}. 
In this work, we address the framework of \emph{$\mathcal{M}$-adhesive categories}. 

\begin{definition}[$\M$-adhesive category]
	A category $\mathcal{C}$ is \emph{$\M$-adhesive} with respect to a class of monomorphisms $\mathcal{M}$ if 
	\begin{itemize}
		\item $\mathcal{M}$ contains all isomorphisms and is \emph{closed under composition and decomposition}, i.e., $f:A \hookrightarrow B,g: B\hookrightarrow C \in \mathcal{M}$ implies $g \circ f \in \mathcal{M}$ 
		and $g \circ f, g \in \mathcal{M}$ implies $f \in \mathcal{M}$.
		\item $\mathcal{C}$ has pushouts and pullbacks along $\mathcal{M}$-morphisms and \emph{$\mathcal{M}$-morphisms are closed under pushouts and pullbacks} such that if Fig.~\ref{fig:po-square} depicts a pushout square with $m \in \mathcal{M}$, then also $n \in \mathcal{M}$, and analogously, if it depicts a pullback square with $n \in \mathcal{M}$, then also $m \in \mathcal{M}$.
		\item Pushouts in $\CC$ along $\mathcal{M}$-morphisms are \emph{vertical weak van Kampen squares}: For any commutative cube as depicted in Fig.~\ref{fig:van-kampen-square} where the bottom square is a pushout along an $\mathcal{M}$-morphism, $b,c,d \in \M$, and the backfaces are pullbacks, then the top square is a pushout if and only if both front faces are pullbacks.
	\end{itemize}
	\vspace{-0.5cm}
	
	\begin{figure}
		\begin{minipage}[b]{.3\textwidth}
			\centering
			\begin{tikzpicture}
				\matrix (m) [	matrix of math nodes,
											row sep=1.25em,
											column sep=1.25em,
											minimum width=1.25em]
				{
					A	&							&	B	\\
						&	\phantom{A}	&	\\
					C	&							&	D	\\};
				\path[-stealth]
					(m-1-1) edge [->] node [left] {\scriptsize $f$} (m-3-1)
									edge [->] node [above] {\scriptsize $m$} (m-1-3)
					(m-3-1) edge [->] node [below] {\scriptsize $n$} (m-3-3)
					(m-1-3) edge [->] node [right] {\scriptsize $g$} (m-3-3);
			\end{tikzpicture}
			\caption{A pushout square.}
			\label{fig:po-square}
		\end{minipage}
		\hfill
		\begin{minipage}[b]{.65\textwidth}
			\centering
			\begin{tikzpicture}
				\matrix (m) [	matrix of math nodes,
											row sep=.35em,
											column sep=.35em,
											minimum width=.35em]
				{
							&							&			&	A^{\prime}	&	\phantom{A}	& \\
					C^{\prime}	&	\phantom{A}	&			&			&							& B^{\prime} \\
							& 						& D^{\prime}	& 		& 						& \\
							&			 				&			&	A		& 						& \\
					C		&							&			&			& 						& B \\
							&							& D		&			& 						& \\};
				\path[-stealth]
					(m-1-4) edge [->] node [above] {\scriptsize $f^{\prime}$} (m-2-1)
									edge [->] node [above] {\scriptsize $m^{\prime}$} (m-2-6)
									edge [->] node [left,near start] {\scriptsize $a$} (m-4-4)
					(m-2-1) edge [->] node [below] {\scriptsize $n^{\prime}$} (m-3-3)
									edge [->] node [left] {\scriptsize $c$} (m-5-1)
					(m-4-4) edge [->] node [above] {\scriptsize $f$} (m-5-1)
									edge [->] node [above] {\scriptsize $m$} (m-5-6)
					(m-5-1) edge [->] node [below] {\scriptsize $n$} (m-6-3)
					(m-5-6) edge [->] node [below] {\scriptsize $g$} (m-6-3)
					(m-2-6) edge [-,draw=white, line width=4pt] (m-3-3)
									edge [->] node [below] {\scriptsize $g^{\prime}$} (m-3-3)
									edge [->] node [right] {\scriptsize $b$} (m-5-6)
					(m-3-3) edge [-,draw=white, line width=4pt] (m-6-3)
									edge [->] node [right,pos=.6] {\scriptsize $d$} (m-6-3);
			\end{tikzpicture}
			\caption{Commutative cube over pushout square.}
			\label{fig:van-kampen-square}
		\end{minipage}
	\end{figure}
\end{definition}
\vspace{-0.5cm}

We write that $(\CC,\M)$ is an $\M$-adhesive category to express that a category $\CC$ is $\M$-adhesive with respect to the class of monomorphisms $\M$ and denote morphisms belonging to $\M$ via a hooked arrow. 
Typical examples of $\M$-adhesive categories are $\mathbf{Set}$ and $\mathbf{Graph}$ (for $\M$ being the class of all injective functions or homomorphisms, respectively). 

\emph{Rules} are used to declaratively describe the \emph{transformation} of objects. 
We use application conditions without introducing \emph{nested conditions} as their formal basis; they are presented in~\LongShort{\cite{HP09} and recalled in Appendix~\ref{sec:nested-conditions}.}{\cite{HP09,KT20}.} 
Moreover, we restrict ourselves to the case of $\M$-matching. 

\begin{definition}[Rules and transformations]\label{def:rules}
	A \emph{rule} $\rho = (p,\mathit{ac})$ consists of a \emph{plain rule} $p$ and an \emph{application condition} $\mathit{ac}$. 
	The plain rule is a span of $\M$-morphisms $p = (L \xhookleftarrow{l} K \xhookrightarrow{r} R)$; the objects are called \emph{left-hand side} (LHS), \emph{interface}, and \emph{right-hand side} (RHS), respectively. 
	The application condition $\mathit{ac}$ is a nested condition over $L$. 
	A \emph{monotonic rule} is a rule, where $l$ is an isomorphism; it is just denoted as $\rho=(r:L \hookrightarrow R,\mathit{ac})$. 
	Given a rule $\rho = (L \xhookleftarrow[]{l} K \xhookrightarrow[]{r} R,\mathit{ac})$ and a morphism $m: L \hookrightarrow G \in \M$, a \emph{(direct) transformation} $G \Rightarrow_{\rho,m} H$ from $G$ to $H$ is given by the diagram in Fig.~\ref{fig:definition-transformation-dpo} where both squares are pushouts and  $m \vDash \mathit{ac}$. 
	If such a transformation exists, the morphism $m$ is called a \emph{match} and rule $\rho$ is \emph{applicable} at match $m$. 
	
	\begin{figure}
		\centering
		\begin{tikzpicture}
			\matrix (m) [	matrix of math nodes,
										row sep=1.25em,
										column sep=1.25em,
										minimum width=1.25em,
										nodes={anchor=center}]
			{
				L	&								&	K	&								& R \\
					&	\phantom{(0)}	&		&	\phantom{(0)}	&	\\
				G & 							& D	& 							& H \\};
			\path[-stealth]
				(m-1-3) edge [left hook->] node [above] {\scriptsize $l$} (m-1-1)
								edge [right hook->] node [above] {\scriptsize $r$} (m-1-5)
								edge [right hook->] node [left] {\scriptsize $d$} (m-3-3)
				(m-1-1) edge [right hook->] node [below,sloped] {\scriptsize $m \vDash \mathit{ac}$} (m-3-1)
				(m-1-5) edge [right hook->] node [right] {\scriptsize $n$} (m-3-5)
				(m-3-3) edge [left hook->] node [below] {\scriptsize $g$} (m-3-1)
								edge [right hook->] node [below] {\scriptsize $h$} (m-3-5);
		\end{tikzpicture}
		\caption{Definition of a direct transformation via two pushouts.}
		\label{fig:definition-transformation-dpo}
	\end{figure}
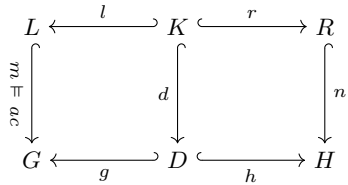
\end{definition}

For some of the following results to hold, we will need $\mathcal{M}$-adhesive categories with further properties. 
\emph{Initial pushouts} are a way to generalize the set-theoretic complement operator categorically. 
\begin{definition}[Boundary and initial pushout]\label{def:initial-pushout}
	Given a morphism $m: L \rightarrow G$ in an $\M$-adhesive category $(\CC,\M)$, an \emph{initial pushout over $m$} is a pushout $(1)$ over $m$ (as depicted in Fig.~\ref{fig:definition-ipo}) such that $b_m \in \M$ and this pushout factors uniquely through every pushout $(3)$ over $m$ where $b_m^{\prime} \in \M$. 
	I.e., for every pushout $(3)$ over $m$ with $b_m^{\prime} \in \M$, there exist unique morphisms $b_m^*,c_m^*$ with $b_m = b_m^{\prime} \circ b_m^*$ and $c_m = c_m^{\prime} \circ c_m^*$. 
	If $(1)$ is an initial pushout, $b_m$ is called \emph{boundary over $m$}, $B_m$ the \emph{boundary object}, and $C_m$ the \emph{context object with respect to $m$}.
	
	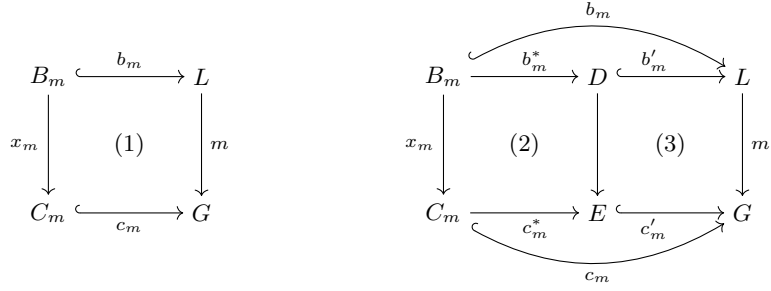
\begin{figure}
		\begin{minipage}{.39\textwidth}
			\centering
			\begin{tikzpicture}
				\matrix (m) [	matrix of math nodes,
											row sep=1.25em,
											column sep=1.25em,
											minimum width=1.25em,
											nodes={anchor=center}]
				{
					B_m	&			& L \\
							&	(1)	&	\\
					C_m & 		&	G \\};
				\path[-stealth]
					(m-1-1) edge [right hook->] node [above] {\scriptsize $b_m$} (m-1-3)
									edge [->] node [left] {\scriptsize $x_m$} (m-3-1)
					(m-1-3) edge [->] node [right] {\scriptsize $m$} (m-3-3)
					(m-3-1) edge [right hook->] node [below] {\scriptsize $c_m$} (m-3-3);
			\end{tikzpicture}
		\end{minipage}
		\hfill
		\begin{minipage}{.59\textwidth}
			\centering
			\begin{tikzpicture}
				\matrix (m) [	matrix of math nodes,
											row sep=1.25em,
											column sep=1.25em,
											minimum width=1.25em,
											nodes={anchor=center}]
				{
					B_m	&			&	D	&			& L \\
							&	(2)	&		&	(3)	&	\\
					C_m & 		& E	& 		&	G \\};
				\path[-stealth]
					(m-1-3) edge [->] (m-3-3)
									edge [right hook->] node [above,pos=.35,inner ysep=2pt] {\scriptsize $b_m^{\prime}$} (m-1-5)
					(m-1-1) edge [->] node [above,pos=.6,inner ysep=2pt] {\scriptsize $b_m^*$} (m-1-3)
									edge [right hook->, bend left] node [above] {\scriptsize $b_m$} (m-1-5)
									edge [->]  node [left] {\scriptsize $x_m$} (m-3-1)
					(m-1-5) edge [->] node [right] {\scriptsize $m$} (m-3-5)
					(m-3-1) edge [->] node [below,pos=.6,inner ysep=2pt] {\scriptsize $c_m^*$} (m-3-3)
									edge [right hook->, bend right] node [below] {\scriptsize $c_m$} (m-3-5)
					(m-3-3) edge [right hook->] node [below,pos=.35,inner ysep=2pt] {\scriptsize $c_m^{\prime}$} (m-3-5);
			\end{tikzpicture}
		\end{minipage}
		\caption{Initial pushout $(1)$ over the morphism $m$ and its factorization property.}
		\label{fig:definition-ipo}
	\end{figure}
\end{definition}
In an $\M$-adhesive category, the square $(2)$ in Fig.~\ref{fig:definition-ipo} is a pushout and $b_m^*,c_m^* \in \M$.  
In $\mathbf{Graph}$, if $m$ is injective, $C_m$ is the minimal completion of $G \setminus m(L)$ (the componentwise set-theoretic difference on nodes and edges) to a subgraph of $G$, and $B_m$ contains the boundary nodes that have to be added for this completion~\cite[Example~6.2]{EEPT06}. 

The existence of \emph{$\M$-effective unions} ensures that the $\M$-subobjects of a given object constitute a lattice. 
\begin{definition}[$\M$-effective unions]\label{def:M-effective-unions}
	An $\M$-adhesive category $(\CC,\M)$ has \emph{$\M$-effective unions} if, for each pushout of a pullback of a pair of $\M$-morphisms, the induced mediating morphism belongs to $\M$ as well. 
\end{definition}

Finally, we recall \emph{$E$-concurrent rules}, which combine the actions of two rules into a single one. 
Their definition assumes a given class $\mathcal{E}^{\prime}$ of pairs of morphisms with the same codomain. 
For the computation of the application condition of a concurrent rule, we refer to~\LongShort{Appendix~\ref{sec:nested-conditions} or \cite{EGHLO14}.}{\cite{EGHLO14,KT20}.}

\begin{definition}[$E$-concurrent rule]\label{def:concurrent-rule}
	Given two rules $\rho_i = (L_i \xhookleftarrow[]{l_i} K_i \xhookrightarrow[]{r_i} R_i, \mathit{ac}_i)$, where $i = 1,2$, an object $E$ with morphisms $e_1: R_1 \to E$ and $e_2:L_2 \to E$ is an \emph{$E$-dependency relation} for $\rho_1$ and $\rho_2$ if $(e_1,e_2) \in \mathcal{E}^{\prime}$ and the pushout complements $(1a)$ and $(1b)$ for $e_1 \circ r_1$ and $e_2 \circ l_2$ 
	(as depicted in Fig.~\ref{fig:E-con-rule}) exist.
	
	\begin{figure}
		\centering
		\includegraphics[width=\textwidth]{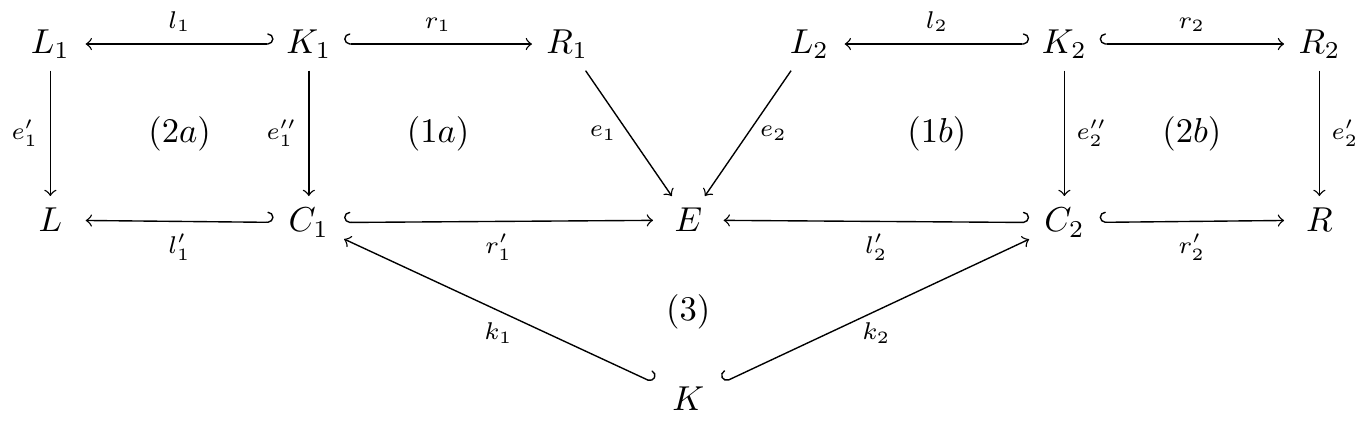}
		\caption{$E$-dependency relation and $E$-concurrent rule.}
		\label{fig:E-con-rule}
	\end{figure}
	
	Given an $E$-dependency relation $E = (e_1,e_2) \in \mathcal{E}^{\prime}$ for rules $\rho_1,\rho_2$, their \emph{$E$-concurrent rule} is defined as
		$\conrule \coloneqq (L \xhookleftarrow[]{l} K \xhookrightarrow[]{r} R, \mathit{ac})$,
	where $l \coloneqq l_1^{\prime} \circ k_1$, $r \coloneqq r_2^{\prime} \circ k_2$, $(1a)$, $(1b)$, $(2a)$, and $(2b)$ are pushouts, $(3)$ is a pullback (also shown in Fig.~\ref{fig:E-con-rule}), and $\mathit{ac}$ is computed in a way that suitably combines the semantics of $\mathit{ac}_1$ and $\mathit{ac}_2$.

	\begin{figure}
		\centering
		\includegraphics[width=\textwidth]{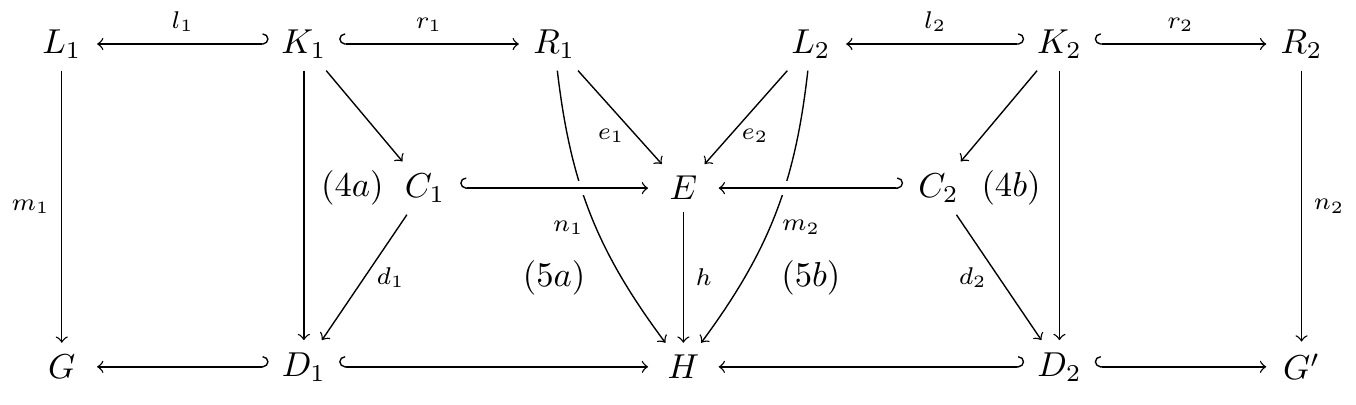}
		\caption{$E$-related transformation.}
		\label{fig:E-related-transformation}
	\end{figure}
	
	A transformation sequence $G \Rightarrow_{\rho_1,m_2} H \Rightarrow_{\rho_2,m_2} \Gprime$ is called \emph{$E$-related} for the $E$-dependency relation $(e_1,e_2) \in \Eprime$ if there exists $h: E \to H$ with $h \circ e_1 = n_1$ and $h \circ e_2 = m_2$ and morphisms $d_i: C_i \to D_i$, where $i= 1,2$, such that $(4a)$ and $(4b)$ commute and $(5a)$ and $(5b)$ are pushouts 
	(see Fig.~\ref{fig:E-related-transformation}).
\end{definition}

The Concurrency Theorem \cite[Theorem~5.23]{EEPT06} states that two $E$-related rule applications may be synthesized into the application of their $E$-concurrent rule and that an application of an $E$-concurrent rule may be analyzed into a sequence of two $E$-related rule applications.

\section{Constructing Generalized Concurrent Rules}
\label{sec:construction-gcr}
In this section, we develop our construction of \emph{generalized concurrent rules} (GCRs) in the context of an $\M$-adhesive category $(\CC,\M)$. 
We first define GCRs and relate them to the construction of concurrent and short-cut rules. 
Subsequently, we elaborate the conditions under which our construction results in a rule and characterize the kinds of rules that are derivable with our construction. 

\subsection{Construction}
Our construction of \emph{generalized concurrent rules} combines the constructions of concurrent and short-cut rules~\cite{FKST18} into a single one. 
It is based on the choice of an $E$-dependency relation as well as of a \emph{common kernel}. 
Intuitively, the $E$-dependency relation captures how both rules are intended to overlap (potentially producing transient elements) whereas the common kernel identifies elements that are deleted by the first rule and recreated by the second one. 
As both concepts identify parts of the interfaces of the involved rules, the construction of a GCR assumes an $E$-dependency relation and a common kernel that are \emph{compatible}. 

\begin{definition}[(Compatible) Common kernel]\label{def:common-kernel_compatibility}
	Given two rules $\rho_i = (L_i \xhookleftarrow{l_i} K_i \xhookrightarrow{r_i} R_i, \mathit{ac}_i)$, where $i = 1,2$, a \emph{common kernel} for them is an $\M$-morphism $k: \Kcap \hookrightarrow V$ with $\M$-morphisms $u_i: \Kcap \hookrightarrow K_i,\ v_1: V \hookrightarrow L_1,\ v_2: V \hookrightarrow R_2$ such that both induced squares 
	$(1a)$ and $(1b)$ in Fig.~\ref{fig:compatibility} are pullbacks. 
	
	\begin{figure}
		\centering
		\begin{tikzpicture}
			\matrix (m) [	matrix of math nodes,
										nodes in empty cells,
										row sep=1.25em,
										column sep=1.25em,
										minimum width={width("$R_2$")},
										nodes={anchor=center}]
			{
				V		&				&			&	&			& \Kcap	&			&	& 		&				& V\\
						& (1a)	&			&	&			&				&			&	&			& (1b)	&	\\	
				L_1	&				& K_1	&	& R_1	&	(2)		& L_2	& & K_2 &				& R_2 \\	
						&				&			&	&			&				&			&	&			&				& \\
						&				&			&	&			& E			&			&	& 		&				& \\};
			\path[-stealth]
				(m-1-1) 		edge [right hook->] node [left] {\scriptsize $v_1$} (m-3-1)
				(m-1-6) 		edge [left hook->] node [above] {\scriptsize $u_1$} (m-3-3.45)
										edge [right hook->] node [above] {\scriptsize $u_2$} (m-3-9.135)
										edge [left hook->] node [above] {\scriptsize $k$} (m-1-1)
										edge [right hook->] node [above] {\scriptsize $k$} (m-1-11)
				(m-1-11)		edge [right hook->] node [right] {\scriptsize $v_2$} (m-3-11) 
				(m-3-3) 		edge [left hook->] node [above] {\scriptsize $l_1$} (m-3-1)
										edge [right hook->] node [above] {\scriptsize $r_1$} (m-3-5)
				(m-3-3.315)	edge [right hook->] node [below,sloped] {\scriptsize $e_1 \circ r_1$} (m-5-6)
				(m-3-5) 		edge [right hook->] node [left] {\scriptsize $e_1$} (m-5-6)
				(m-3-7) 		edge [left hook->] node [right] {\scriptsize $e_2$} (m-5-6)
				(m-3-9) 		edge [left hook->] node [above] {\scriptsize $l_2$} (m-3-7)
										edge [right hook->] node [above] {\scriptsize $r_2$} (m-3-11)
				(m-3-9.225)	edge [left hook->] node [below,sloped] {\scriptsize $e_2 \circ l_2$} (m-5-6);
		\end{tikzpicture}
		\caption{Compatibility of common kernel with $E$-dependency relation.}
		\label{fig:compatibility}
	\end{figure}
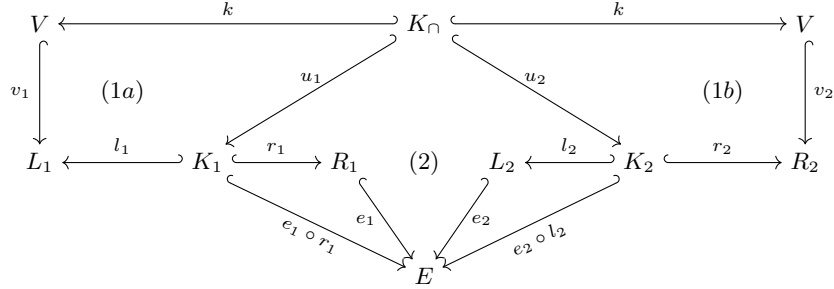
	
	Given additionally an $E$-dependency relation $E = (e_1: R_1 \hookrightarrow E, e_2:L_2 \hookrightarrow E) \in \Eprime$ for $\rho_1$ and $\rho_2$, $E$ and $k$ are \emph{compatible} if square $(2)$ 
	is a pullback.  
\end{definition}

In the following, we will often suppress the morphisms $u_i,v_i$ from our notation and just speak of a common kernel $k: \Kcap \hookrightarrow V$. 
As an $\M$-morphism $k$ might constitute a common kernel for a pair of rules in different ways, we implicitly assume the embedding to be given.

\begin{example}\label{ex:common-kernel}
	Figure~\ref{fig:example-common-kernel} shows an $E$-dependency relation and a common kernel compatible with it for rules \emph{removeMiddleMan} and \emph{extractSubclass} (Fig.~\ref{fig:rules-class-diagram}). 
	The names of the nodes also indicate how the morphisms are defined. 
	The concurrent rule for this $E$-dependency relation is the rule \emph{Ref2Gen\_CR} in Fig.~\ref{fig:rules-class-diagram}. 
	
	\begin{figure}%
		\centering
		\includegraphics[width=\textwidth]{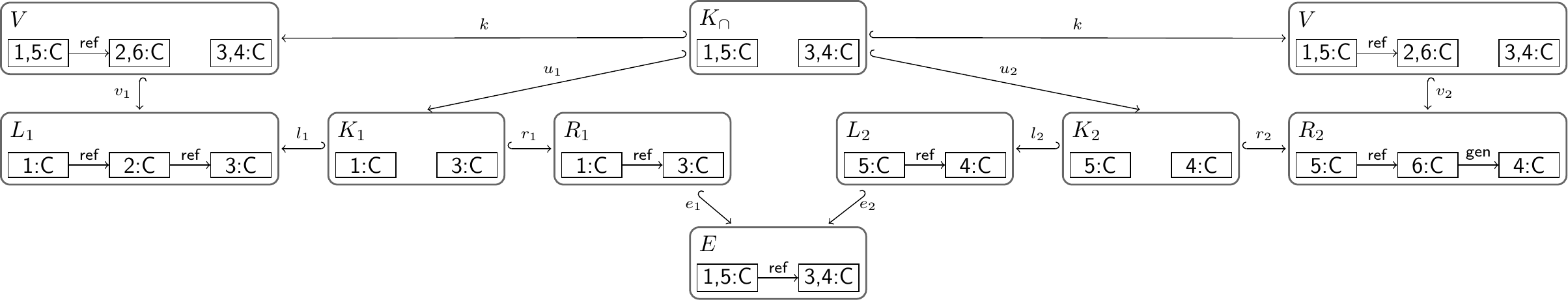}%
		\caption{Common kernel for rules \emph{removeMiddleMan} and \emph{extractSubclass}.}
		\label{fig:example-common-kernel}
	\end{figure}
\end{example}

The following lemma is the basis for the construction of generalized concurrent rules; it directly follows from the definition of the interface $K$ of a concurrent rule as pullback.  

\begin{lemma}\label{lem:existence-extension-morphism}
	Given two rules $\rho_1,\rho_2$, an $E$-dependency relation $E$, and a common kernel $k$ for  $\rho_1,\rho_2$ that is compatible with $E$, there exists a unique $\M$-morphism $p:\Kcap \hookrightarrow K$, where $K$ is the interface of the concurrent rule $\rho_1 *_E \rho_2$, such that $k_i \circ p = e_i^{\prime\prime} \circ u_i$ for $i=1,2$ (compare the diagrams in Definitions~\ref{def:concurrent-rule} and \ref{def:common-kernel_compatibility}). 
\end{lemma}

A GCR extends a concurrent rule by enhancing its interface $K$ with the additional elements in $V$ of a given common kernel. 
Formally, this means to compute a pushout along the just introduced morphism $p$. 

\begin{construction}\label{con:span-gen-con-rule-1}
	Given two plain rules $\rho_i = (L_i \xhookleftarrow{l_i} K_i \xhookrightarrow{r_i} R_i)$, where $i=1,2$, an $E$-dependency relation $E = (e_1: R_1 \hookrightarrow E, e_2: L_2 \hookrightarrow E) \in \mathcal{E}^{\prime}$, and a common kernel $k: \Kcap \hookrightarrow V$ of $\rho_1$ and $\rho_2$ that is compatible with $E$, we construct the span $L \xleftarrow{l^{\prime}} K^{\prime} \xrightarrow{r^{\prime}} R$ as follows (compare Fig.~\ref{fig:construction-gcr-new}):
	\begin{enumerate}
		\item Compute the concurrent rule $\rho_1 *_E \rho_2 = (L \xhookleftarrow{l} K \xhookrightarrow{r} R)$. 
		\item Compute $\Kprime$ as pushout of $k$ along $p$, where $p: \Kcap \hookrightarrow K$ is the unique morphism existing according to Lemma~\ref{lem:existence-extension-morphism} (depicted twice in Fig.~\ref{fig:construction-gcr-new}).
		\item The morphism $\lprime: \Kprime \to L$ is the unique morphism with $\lprime \circ \pprime = e_1^{\prime} \circ v_1$ and $\lprime \circ \kprime = l_1^{\prime} \circ k_1$ that is induced by the universal property of the pushout computing $\Kprime$. 
		The morphism $\rprime$ is defined analogously.
	\end{enumerate}
	\begin{figure}
		\centering
		\includegraphics[width=\textwidth]{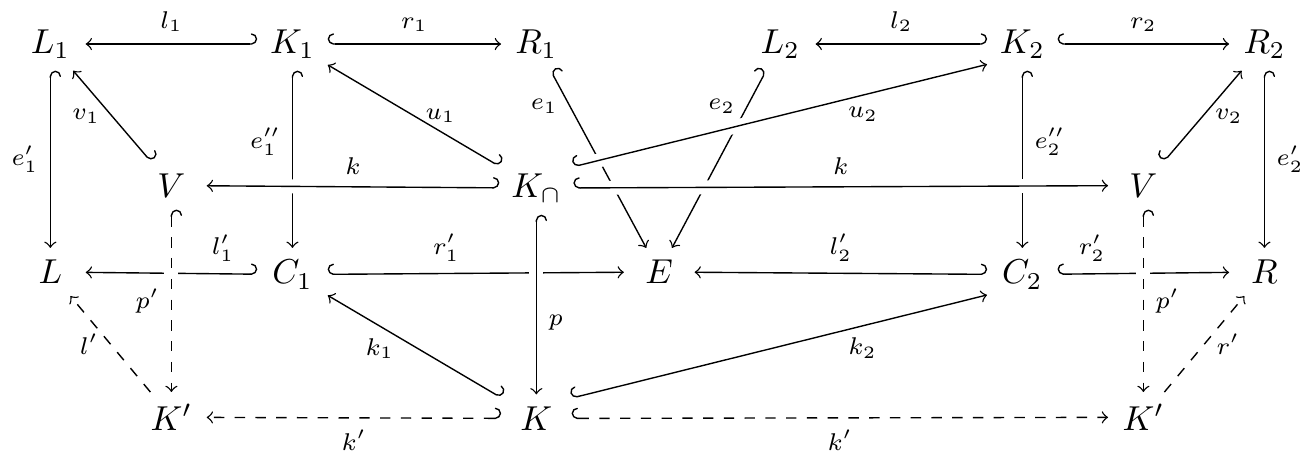}%
		\caption{Construction of a generalized concurrent rule.}%
		\label{fig:construction-gcr-new}%
	\end{figure}
\end{construction}

\begin{definition}[Generalized concurrent rule. Enhancement morphism]\label{def:generalized-concurrent-rule}
	Given two rules $\rho_1,\rho_2$, an $E$-dependency relation $E$, and a common kernel $k$ of $\rho_1$ and $\rho_2$ that are compatible such that the span $L \xleftarrow{l^{\prime}} K^{\prime} \xrightarrow{r^{\prime}} R$ obtained from Construction~\ref{con:span-gen-con-rule-1} consists of $\M$-morphisms, the \emph{generalized concurrent rule} of $\rho_1$ and $\rho_2$, given $E$ and $k$, is defined as $\genconrule \coloneqq (L \xhookleftarrow{l^{\prime}} K^{\prime} \xhookrightarrow{r^{\prime}} R, \mathit{ac})$ with $\mathit{ac}$ being the application condition of the concurrent rule $\rho_1 *_E \rho_2 = (L \xhookleftarrow{l} K \xhookrightarrow{r} R, \mathit{ac})$.
	
	The unique $\M$-morphism $\kprime: K \hookrightarrow \Kprime$ with $l = l_1^{\prime} \circ k_1 = \lprime \circ \kprime$ and $r = r_2^{\prime} \circ k_2 = \rprime \circ \kprime$, which is obtained directly from the construction, is called \emph{enhancement morphism}. 
	We also say that $\genconrule$ is a GCR \emph{enhancing} $\conrule$. 
\end{definition}

\begin{example}
	\emph{Ref2Gen\_GCR} is a GCR that enhances \emph{Ref2Gen\_CR} (Fig.~\ref{fig:rules-class-diagram}); it is constructed using the common kernel presented in Example~\ref{ex:common-kernel}.
	Figure~\ref{fig:example-construction-gcr} illustrates the computation of its interface $\Kprime$ and left-hand morphism $\lprime$. 
	The pushout of $k$ and $p$ extends the interface of \emph{Ref2Gen\_CR} by the \textsf{Class 2,6} and its incoming reference. 
	
	\begin{figure}[t]
		\centering
		\includegraphics[scale=.7]{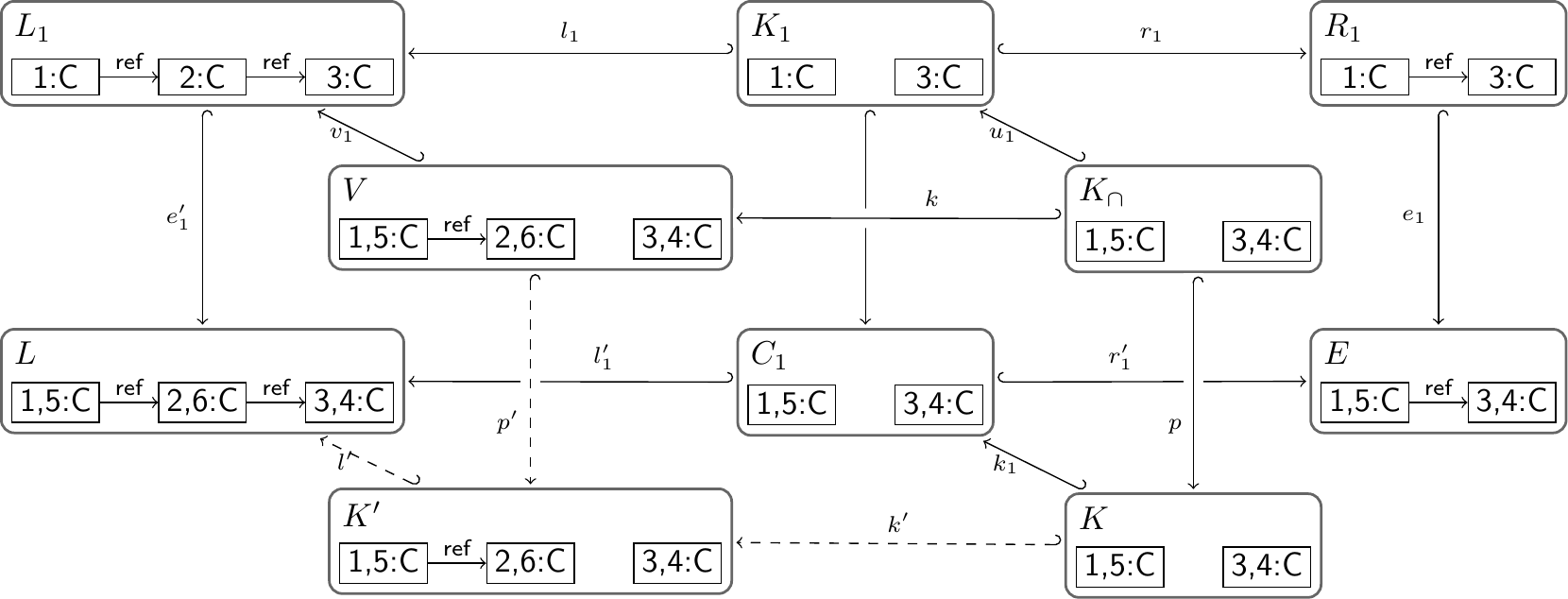}%
		\caption{Computing the interface and the left-hand morphism of \emph{Ref2Gen\_GCR}.}%
		\label{fig:example-construction-gcr}%
	\end{figure}
\end{example}

\begin{note}[Assumptions and notation]
	For the rest of the paper, we fix the following assumptions:
	We work in an $\M$-adhesive category $(\CC,\M)$ with a given class $\mathcal{E}^{\prime}$ of pairs of morphisms with the same codomain such that $\CC$ possesses an $\Eprime$-$\M$ pair factorization.\footnote{%
	This means, every pair of morphisms with the same codomain can be factored as a pair of morphisms belonging to $\Eprime$ followed by an $\M$-morphism.
	We do not directly need this property in any of our proofs but it is assumed for the computation of application conditions of concurrent and, hence, also generalized concurrent rules.
	Moreover, it guarantees the existence of $E$-related transformations~\cite[Fact~5.29]{EEGH15}. 
	
	Since we restrict ourselves to the case of $\M$-matching, decomposition of $\M$-morphisms then ensures that all occurring pairs $(e_1,e_2) \in \mathcal{E}^{\prime}$ are in fact even pairs of $\M$-morphisms. 
	This in turn (by closedness of $\M$ under pullbacks) implies that in any common kernel $k$ compatible to a given $E$-dependency relation, the embedding morphisms $u_1,u_2$ are necessarily $\M$-morphisms.
}
	Further categorical assumptions are mentioned as needed.  
	Furthermore, we always assume two rules $\rho_1,\rho_2$, an $E$-dependency relation $E$, and a common kernel $k$ for them to be given such that $E$ and $k$ are compatible. 
	We consistently use the notations and names of morphisms as introduced above (and in Fig.~\ref{fig:construction-gcr-new}). 
\end{note}

\subsection{Relating Generalized Concurrent Rules to other Kinds of Rules}
\label{sec:relating-gcrs}
In this section, we relate generalized concurrent rules to other variants of rule composition. 

\paragraph{Concurrent Rules} are the established technique of sequential rule composition in double-pushout rewriting. 
By definition, the left- and right-hand sides of a GCR coincide with the ones of the concurrent rule it enhances. 
One also directly obtains that a GCR coincides with its underlying concurrent rule if and only if its common kernel $k$ is chosen to be an isomorphism. 

\begin{proposition}[A concurrent rule is a GCR]\label{prop:con-rule-as-gen-con-rule}
	Given a concurrent rule $\conrule$ and a GCR $\genconrule$ enhancing it, the enhancement morphism $\kprime: K \hookrightarrow \Kprime$ is an isomorphism if and only if $k$ is one. 
	In particular, $\conrule$ coincides with $\genconrule$ (up to isomorphism) for $k = id_{\Kcap}$, where $\Kcap$ is obtained by pulling back $(e_1 \circ r_1, e_2 \circ l_2)$. 
\end{proposition}

\paragraph{Short-cut rules}~\cite{FKST18} 
are a further, very specific kind of sequentially composed rules (for the definition of which we refer to \LongShort{Definition~\ref{def:scrule} in Appendix~\ref{sec:short-cut-rules}).}{\cite{FKST18,KT20}).} 
In an adhesive category, given a rule that only deletes and a rule that only creates, a short-cut rule combines their sequential effects into a single rule that allows to identify elements that are deleted by the first rule as recreated by the second and to preserve them instead. 
The construction of GCRs we present here now fuses our construction of short-cut rules with the concurrent rule construction. 
This means, we lift that construction from its very specific setting (adhesive categories and monotonic, plain rules) to a far more general one ($\M$-adhesive categories and general rules with application conditions). 
This is of practical relevance as, in application-oriented work on incremental model synchronization, we are already employing short-cut rules in more general settings (namely, we compute short-cut rules from monotonic rules with application conditions rewriting typed attributed triple graphs, which constitute an adhesive HLR category that is not adhesive)~\cite{FKST19,FKST20,FKMST20}. 

\begin{proposition}[A short-cut rule is a GCR]\label{prop:sc-rule-is-gcr}
	Let $\CC$ be an adhesive category and the class $\mathcal{E}^{\prime}$ be such that it contains all pairs of jointly epic $\M$-morphisms.
	Let $r_i = (r_i: L_i \hookrightarrow R_i)$, where $i = 1,2$, be two monotonic rules and $k: \Kcap \hookrightarrow V$ a common kernel for them. 
	Then the short-cut rule $\scrule{r_1}{r_2}{k}$ coincides with the generalized concurrent rule $r_1^{-1} *_{E,k} r_2$, where $E = (e_1,e_2)$ is given via pushout of $(u_1,u_2)$. 
\end{proposition}

\paragraph{Parallel and amalgamated rules} are further kinds of rules arising by composition. 
Whereas concurrent rules combine the sequential application of two rules, an amalgamated rule combines the application of two (or more) rules to the same object into the application of a single rule~\cite{BFH85,GHE14}.
In categories with coproducts, the parallel rule is just the sum of two rules; for plain rules (i.e., without application conditions) it is a special case of the concurrent as well as of the amalgamated rule construction. 
A thorough presentation of all three forms of rule composition in the context of $\M$-adhesive categories, rules with application conditions, and general matching can be found in~\cite{EGHLO14}. 
When introducing short-cut rules~\cite{FKST18}, we showed that their effect cannot be achieved by concurrent or amalgamated rules. 
Thus, by the above proposition, the same holds for GCRs; they indeed constitute a new form of rule composition. 
The relations between the different kinds of rule composition are summarized in Fig.~\ref{fig:relations-rule-constructions}. 
The lines from parallel rule are dashed as the indicated relations only hold in the absence of application conditions and, for short-cut rules, in the specific setting only in which these are defined.
\begin{figure}[t]
	\centering
	\includegraphics{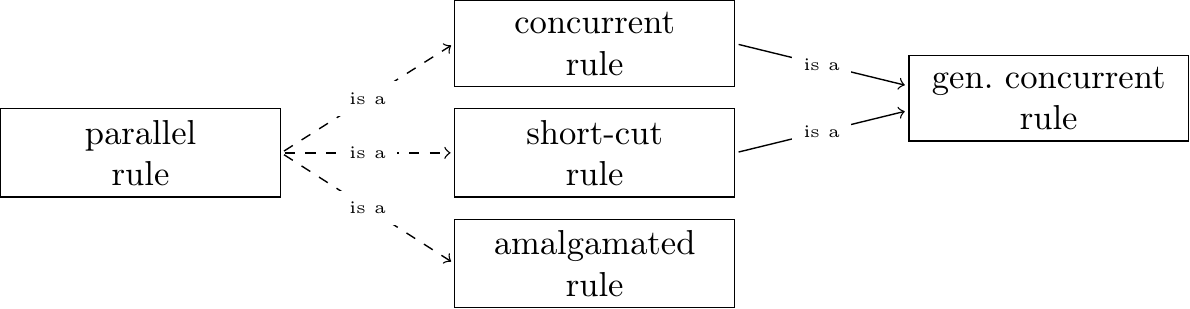}%
	\caption{Relations between the different kinds of rule composition.}%
	\label{fig:relations-rule-constructions}%
\end{figure}

\subsection{Characterizing Derivable Generalized Concurrent Rules}
\label{sec:characterization-derivable-rules}
Next, we characterize the GCRs derivable from a given pair of rules. 
We do so in two different ways, namely (i) characterizing possible choices for the morphisms $v_1,v_2$ in a common kernel and (ii) characterizing the possible choices for enhancement morphisms $\kprime$. 

\begin{proposition}[Embedding characterization of GCRs]\label{prop:embedding-characterization-monic}
	Let $(\CC,\M)$ be an $\M$-adhesive category with $\M$-effective unions. 
	\begin{enumerate}
		\item The application of Construction~\ref{con:span-gen-con-rule-1} results in a GCR $\genconrule$ if and only if $v_1,v_2 \in \M$. 
		\item The assumption that $\M$-effective unions exist is necessary for this result to hold. 
	\end{enumerate}
\end{proposition}

Next, we consider enhancement morphisms in more detail and answer the following question:
Given a concurrent rule $\rho_1 *_E \rho_2 = L \hookleftarrow K \hookrightarrow R$, which common $\M$-subobjects $\Kprime$ of $L$ and $R$ that enhance $K$ constitute the interface of a GCR? 
The following example shows that not all of them do. 

\begin{example}
In the category $\mathbf{Graph}$ (or $\mathbf{Set}$), if $p_1$ is the trivial rule $\emptyset \hookleftarrow \emptyset \hookrightarrow \emptyset$ and $p_2 = (\cnode \hookleftarrow \emptyset \hookrightarrow \cnode)$, it is straightforward to verify that $p_2$ can be derived as concurrent rule again (for the $E$-dependency object $E = \cnode$). 
However, $\cnode \hookleftarrow \cnode \hookrightarrow \cnode$ cannot be derived as GCR from these two rules since  $p_1$ does not delete a node.
\end{example}

It turns out that only elements that are deleted by the first rule and created by the second can be identified and incorporated into $\Kprime$.  
The next proposition clarifies this connection using the language of initial pushouts.

\begin{definition}[Appropriately enhancing]\label{def:appropriately-enhancing}
	In a category $\CC$ with initial push\-outs, let $\conrule = L \xhookleftarrow{l} K \xhookrightarrow{r} R$ be an $E$-concurrent rule and $\kprime: K \hookrightarrow \Kprime$ be an $\M$-morphism such that there exist $\M$-morphisms $\lprime: \Kprime \hookrightarrow L$ and $\rprime: \Kprime \hookrightarrow R$ with $\lprime \circ \kprime = l$ and $\rprime \circ \kprime = r$. 
	Then $\kprime$ is called \emph{appropriately enhancing} if the following holds (compare Fig.~\ref{fig:definition-appropriate-enhancement}): 
	The boundary and context objects $B_{\kprime}$ and $C_{\kprime}$ of the initial pushout over $\kprime$ factorize via $\M$-morphisms $s_L,s_R:B_{\kprime} \hookrightarrow B_{l_1},B_{r_2}$ and $t_L,t_R:C_{\kprime} \hookrightarrow C_{l_1},C_{r_2}$ as pullback through the initial pushouts over $l_1: K_1 \hookrightarrow L_1$ and $r_2: K_2 \hookrightarrow R_2$ in such a way that $k_1 \circ b_{\kprime} = e_1^{\prime\prime} \circ b_{l_1} \circ s_L$ and $k_2 \circ b_{\kprime} = e_2^{\prime\prime} \circ b_{r_2} \circ s_R$. 
	
	\begin{figure}%
		\centering
		\begin{tikzpicture}
			\matrix (m) [	matrix of math nodes,
										nodes in empty cells,
										nodes={anchor=center},
										row sep=.5em,
										column sep=.5em,
										minimum width=.5em]
			{
				C_1	&									&					&								&	K						&								& 				&									& C_2 \\
						&									&					&								&							&								&					&									& \\
				K_1	&									&	B_{l_1}	&								& B_{\kprime}	&								& B_{r_2}	&									& K_2 \\
						& (\mathit{IPO})	&					& (\mathit{PB})	&							& (\mathit{PB})	& 				& (\mathit{IPO})	& \\
				L_1	&									&	C_{l_1}	&								&	C_{\kprime}	&								& C_{r_2}	&									& R_2 \\};
			\path[-stealth]
				(m-3-1) edge [right hook->] node [left] {\scriptsize $l_1$} (m-5-1)
				(m-3-1) edge [left hook->] node [left] {\scriptsize $e_1^{\prime\prime}$} (m-1-1)
				(m-3-3) edge [right hook->] node [left,inner xsep=2pt] {\scriptsize $x_{l_1}$} (m-5-3)
								edge [left hook->] node [above,pos=.4] {\scriptsize $b_{l_1}$} (m-3-1)
				(m-3-5) edge [right hook->] node [left,inner xsep=2pt] {\scriptsize $x_{\kprime}$} (m-5-5)
								edge [left hook->] node [above] {\scriptsize $s_L$} (m-3-3)
								edge [right hook->] node [above] {\scriptsize $s_R$} (m-3-7)			
				(m-3-7) edge [right hook->] node [left,inner xsep=2pt] {\scriptsize $x_{r_2}$} (m-5-7)
								edge [right hook->] node [above,pos=.4] {\scriptsize $b_{r_2}$} (m-3-9)
				(m-3-9) edge [right hook->] node [right] {\scriptsize $r_2$} (m-5-9)
				(m-3-9)	edge [left hook->] node [right] {\scriptsize $e_2^{\prime\prime}$} (m-1-9)
				(m-5-3) edge [left hook->] node [above] {\scriptsize $c_{l_1}$} (m-5-1)
				(m-5-5) edge [left hook->] node [above] {\scriptsize $t_L$} (m-5-3)
								edge [right hook->] node [above] {\scriptsize $t_R$} (m-5-7)
				(m-5-7) edge [right hook->] node [above] {\scriptsize $c_{r_2}$} (m-5-9)
				(m-3-5) edge [left hook->] node [left,inner xsep=2pt] {\scriptsize $b_{\kprime}$} (m-1-5)
				(m-1-5) edge [left hook->] node [above] {\scriptsize $k_1$} (m-1-1)
								edge [right hook->] node [above] {\scriptsize $k_2$} (m-1-9);
		\end{tikzpicture}
		\caption{Definition of \emph{appropriate enhancement}.}
		\label{fig:definition-appropriate-enhancement}	
	\end{figure}
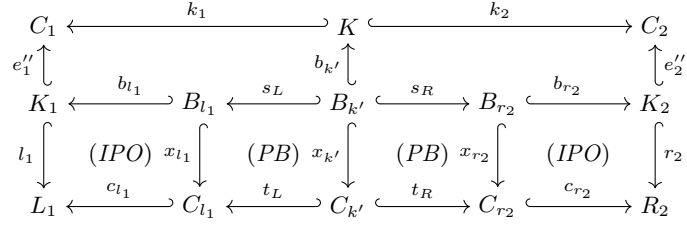
\end{definition}

On the level of graph elements (in fact: arbitrary categories of presheaves over $\mathbf{Set}$), this means the following: 
When considering $\Kprime$ as a subobject of $L$ via $\lprime$, the elements of $\Kprime \setminus K$ have to be mapped to elements of $L_1 \setminus K_1$.  
When considering $\Kprime$ as a subobject of $R$ via $\rprime$, the elements of $\Kprime \setminus K$ have to be mapped to elements of $R_2 \setminus K_2$. 

\begin{example}
	Figure~\ref{fig:example-appropriately-enhancing} illustrates the notion of appropriate enhancement using our running example. 
	The two inner squares constitute pullbacks, which means that the additional elements of $\Kprime$, namely \textsf{Class 2,6} and its incoming reference, are mapped to elements deleted by \emph{removeMiddleMan} via $t_L$ and to elements created by \emph{extractSubclass} via $t_R$. 
	Moreover, for both $C_1$ and $C_2$ the two possible ways to map \textsf{Class 1,5} from $B_{\kprime}$ to it coincide. 
	
	\begin{figure}
		\centering
		\includegraphics[width=\linewidth]{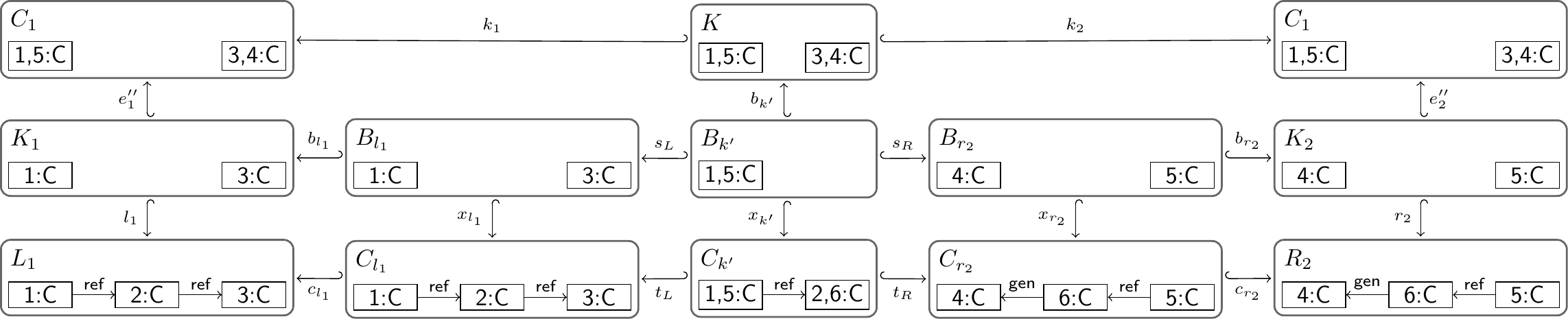}
		\caption{Illustrating the property of appropriate enhancement.}
		\label{fig:example-appropriately-enhancing}
	\end{figure}
\end{example}

\begin{proposition}[Enhancement characterization of GCRs]\label{prop:enhancement-characterization-monic}
	Let $(\CC,\M)$ be an $\M$-adhesive category with initial pushouts. 
	Given an $E$-concurrent rule $\conrule = L \xhookleftarrow{l} K \xhookrightarrow{r} R$ and an $\M$-morphism $\kprime: K \hookrightarrow \Kprime$ such that there exist $\M$-morphisms $\lprime: \Kprime \hookrightarrow L$ and $\rprime: \Kprime \hookrightarrow R$ with $\lprime \circ \kprime = l$ and $\rprime \circ \kprime = r$, the span $L \xhookleftarrow{\lprime} \Kprime \xhookrightarrow{\rprime} R$ is derivable as a GCR $\genconrule$ if and only if $\kprime$ is appropriately enhancing. 
\end{proposition}

In $\mathbf{Graph}$, the result above characterizes a GCR as a rule whose interface $\Kprime$ enhances the interface $K$ of the enhanced concurrent rule by identifying elements of $L_1 \setminus K_1$ and $R_2 \setminus K_2$ with each other and including them in $\Kprime$. 

\begin{corollary}[Enhancement characterization in the category $\mathbf{Graph}$]\label{cor:derivable-gcrs-graph}
	In the category of (typed/attributed) graphs, given an $E$-dependency relation $E$ for two rules, every GCR $\genconrule$ enhancing $\conrule$ is obtained in the following way:
	$\Kprime$ arises by adding new graph elements to $K$ and $\lprime$ and $\rprime$ extend the morphisms $l$ and $r$ in such a way that they (i) remain injective graph morphisms and (ii) under $\lprime$ the image of every newly added element is in $L_1 \setminus K_1$ and in $R_2 \setminus K_2$ under $\rprime$.
	
	Assuming finite graphs only, the number of GCRs enhancing a concurrent rule $\conrule$ may grow factorially in $\min(|L_1 \setminus K_1|,|R_2 \setminus K_2|)$. 
\end{corollary}

\section{A Generalized Concurrency Theorem}
\label{sec:generalized-concurrency-theorem}
In this section, we present our Generalized Concurrency Theorem that clarifies how sequential applications of two rules relate to an application of a GCR derived from them. 
As a prerequisite, we present a proposition that relates applications of concurrent rules with those of enhancing GCRs. 
It states that an application of a GCR leads to the same result as one of the enhanced concurrent rule; however, by its application, more elements are preserved (instead of being deleted and recreated). 

\begin{proposition}[Preservation property of GCRs]\label{prop:preservation-property}
	Let $G_0 \Rightarrow_{\conrule,m} G_2$ be a transformation via concurrent rule $\conrule$, given by the span $G_0 \xhookleftarrow{g_0} D \xhookrightarrow{g_2} G_2$. For any GCR $\genconrule$ enhancing $\conrule$, there is a transformation $G_0 \Rightarrow_{\genconrule,m} G_2$, given by a span $G_0 \xhookleftarrow{g_0^{\prime}} D^{\prime} \xhookrightarrow{g_2^{\prime}} G_2$, and a unique $\M$-morphism $\kprimeprime: D \hookrightarrow D^{\prime}$ such that $g_i^{\prime} \circ \kprimeprime = g_i$ for $i= 0,2$. 
	Moreover, $\kprimeprime$ is an isomorphism if and only if the enhancement morphism $\kprime$ is one. 
\end{proposition}

\begin{example}
	Every match for the concurrent rule \emph{Ref2Gen\_CR} is also one for the GCR \emph{Ref2Gen\_GCR}. 
	Moreover, the two results of the according applications will be isomorphic. 
	The above proposition, however, formally captures that the application of \emph{Ref2Gen\_CR} will delete more elements than the one of \emph{Ref2Gen\_GCR}. 
	The graph intermediately arising during the application of the former (not containing the class to which \textsf{Class 2} had been matched) properly embeds into the one arising during the application of the latter. 
\end{example}

The classical Concurrency Theorem states that a sequence of two rule applications can be 
replaced by an application of a concurrent rule (synthesis) and, vice versa (analysis). 
The synthesis is still possible in the case of GCRs. 
The analysis, however,  holds under certain conditions only.  
Next, we illustrate how the analysis might fail.  
Subsequently, we state the Generalized Concurrency Theorem.

\begin{example}\label{ex:counterexample-analysis-3}
	When \emph{Ref2Gen\_GCR} is applied at a match that maps \textsf{Class 2,6} to a node with incoming or outgoing references or generalizations beyond the two references required by the match, the dangling edge condition prevents the applicability of the underlying first rule \emph{removeMiddleMan} at the induced match. 
	The deletion of \textsf{Class 2} would not be possible because of these additional adjacent edges. 
	Hence, the analysis of the application of \emph{Ref2Gen\_GCR} into sequential applications of \emph{removeMiddleMan} and \emph{extractSubclass} fails in that situation. 
\end{example}

\begin{theorem}[Generalized Concurrency Theorem]\label{thm:generalized-concurrency-theorem} 
	Let $(\CC,\M)$ be an $\M$-adhesive category with $\M$-effective unions. 
	
\noindent \textbf{Synthesis.} For each $E$-related transformation sequence $G_0 \Rightarrow_{\rho_1,m_1} G_1 \Rightarrow_{\rho_2,m_2} G_2$ there exists a direct transformation $G_0 \Rightarrow_{\genconrule,m} G_2$. 

\noindent \textbf{Analysis.} Given a direct transformation $G_0 \Rightarrow_{\genconrule,m} G_2$, there exists an $E$-related transformation sequence $G_0 \Rightarrow_{\rho_1,m_1} G_1 \Rightarrow_{\rho_2,m_2} G_2$ with $m_1 = m \circ e_1^{\prime}$ if and only if $G_0 \Rightarrow_{\rho_1,m_1} G_1$ exists, i.e., if and only if $\rho_1$ is applicable at $m_1$.
\end{theorem}

\begin{remark}
	At least for the case of $\M$-matching and in the presence of $\M$-effective unions, the Concurrency Theorem indeed becomes a corollary to our Generalized Concurrency Theorem. 
	This is due to the observation that a GCR is a concurrent rule if and only if $k$ is an isomorphism (Proposition~\ref{prop:con-rule-as-gen-con-rule}). 
	Similarly, the Generalized Concurrency Theorem subsumes the Short-cut Theorem~\cite[Theorem~7]{FKST18}. 
	
	In the case of graphs, the Generalized Concurrency Theorem ensures that the situation illustrated in Example~\ref{ex:counterexample-analysis-3} is the only situation in which the analysis of the application of a GCR fails: 
	When restricting to injective matching, a violation of the dangling edge condition is known to be the only possible obstacle to an application of a graph transformation rule~\cite[Fact~3.11]{EEPT06}. 
\end{remark}

\section{Related Work}
\label{sec:related-work}
In this paper, we present a construction for generalized concurrent rules (GCRs) based on the double-pushout approach to rewriting.  We compare it with existing constructions of concurrent rules that use some categorical setting. 

Concerning double-pushout rewriting, after presenting Concurrency Theorems in specific categories (such as \cite{ER80} for the case of graphs), such a theorem in a rather general categorical setting was obtained by Ehrig et al.~\cite{EHKP91}. 
In that work, spans $R_1 \leftarrow D \rightarrow L_2$ are used to encode the information about rule dependency. 
After (variants of) adhesive categories had been established as an axiomatic basis for double-pushout rewriting~\cite{LS05}, the construction of concurrent rules and an according Concurrency Theorem was lifted to that setting: 
directly in~\cite{LS05} with the dependency information still encoded as span and by Ehrig et al.~\cite{EEPT06} with the dependency information now encoded as a co-span $R_1 \rightarrow E \leftarrow L_2$.
Finally, the last construction, addressing plain rules, has been extended to the case of rules with general application conditions~\cite{EHL10}. 
It is this construction that we present in Definition~\ref{def:concurrent-rule}.  
A generalization of this construction to enable the reuse of graph elements (or object parts in general) as we present it in this paper is new. 

Sequential composition of rules has also been presented for other categorical approaches to rewriting, for example, for \emph{single-pushout rewriting}~\cite{Loewe93}, \emph{sesqui-pushout rewriting}~\cite{Loewe15,Behr19}, or for \emph{double-pushout rewriting in context} (at least for special cases)~\cite{Loewe19}. 
The theory is most mature in the sesqui-pushout case, where Behr~\cite{Behr19} has established a construction of concurrent rules in the setting of $\M$-adhesive categories for rules with general application conditions.  
It seems that our construction of GCRs would be similarly applicable to sesqui-pushout rewriting; however, applied as is, it would have a restricted expressivity in that context: 
In the category $\mathbf{Graph}$, for instance, a rule application according to the sesqui-pushout semantics implies to (implicitly) delete all edges that are incident to a deleted node. 
Our construction would have to be extended for being able to identify such implicitly deleted items such that they can be preserved. 

In the \emph{cospan DPO approach}, rules are cospans (instead of spans) and the order of computation is switched compared to classical DPO rewriting, i.e., creation precedes deletion of elements~\cite{EHP09}. 
This approach has been used, for example, to simultaneously rewrite a model and its meta-model~\cite{MTLW15} or to formalize the rewriting of graphs that are typed over a chain of type graphs~\cite{WMR20} (here actually as cospan sesqui-pushout rewriting).    
As creation precedes deletion, the cospan DPO approach intrinsically offers support for certain kinds of information preservation. 
For example, attribute values of nodes that are to be deleted can be passed to newly created nodes first. 
However, in the category of graphs, cospan DPO rewriting is subject to virtually the same dangling edge condition as classical DPO rewriting (see~\cite{EHP09}). 
This means, employing the cospan DPO approach instead of classical DPO rewriting does not address the problem of the applicability of rules. 
Moreover, modeling the kind of information preservation we are interested in (regarding elements deleted by one rule as recreated by a second) would require a specific form of rule composition also in the cospan approach. 
We do not expect this to be essentially simpler than the construction we provide for the classical DPO approach in this paper. 

Behr and Sobociński~\cite{BS20} proved that, in the case of $\M$-matching, the concurrent rule construction is associative (in a certain technical sense; the same holds true for the sesqui-pushout case~\cite{Behr19}). 
Based on that result, they presented the construction of a rule algebra that captures interesting properties of a given grammar. 
This has served, for example, as a starting point for a static analysis of stochastic rewrite systems~\cite{BDG20}. 
Considering our GCR construction, it is future work to determine whether it is associative as well.  

Kehrer et al.~\cite{KTRK16} addressed the automated generation of edit rules for models based on  a given meta-model. 
Besides basic rules that create or delete model elements, they also generated \emph{move} and \emph{change rules}. 
It turns out that these can be built as GCRs from their basic rules but not as mere concurrent rules. 
Additionally, we introduced short-cut rules for more effective model synchronization~\cite{FKST19,FKST20,FKMST20} and showed here that short-cut rules are special GCRs. Hence, these works suggest that GCRs can capture typical properties of model edit operators. 
A systematic study of which edit operators can be captured as GCRs (and which GCRs capture typical model edit operators) remains future work.

\section{Conclusion}
\label{sec:conclusion}
In this paper, we present \emph{generalized concurrent rules} (GCRs) as a construction that generalizes the constructions of concurrent and short-cut rules. 
We develop our theory of GCRs in the setting of double-pushout rewriting in $\M$-adhesive categories using rules with application conditions applied at $\M$-matches only. 
In contrast to concurrent rules, GCRs allow reusing elements that are deleted in the first rule and created in the second. 
As a central result (Theorem~\ref{thm:generalized-concurrency-theorem}), we generalize the classical Concurrency Theorem. 

From a theoretical point of view, it would be interesting to develop similar kinds of rule composition in the context of other categorical approaches to rewriting like the single- or the sesqui-pushout approach. 
Considering practical application scenarios, we are most interested in classifying the derivable GCRs of a given pair of rules according to their use and computing those that are relevant for certain applications in an efficient way. 

\paragraph{Acknowledgments} This work was partially funded by the German Research Foundation (DFG), project TA294/17-1.

\bibliographystyle{splncs04}
\bibliography{literature}

\LongShort{%
	\appendix
	
	\section{Additional Preliminaries}
	\label{sec:additional-preliminaries}
	In this section, we collect additional technical preliminaries and concepts needed for our proofs. 
First, in Appendix~\ref{sec:nested-conditions}, we recall \emph{nested conditions} and the \emph{Shift} and \emph{Left} constructions that allow to \enquote{move} conditions~\cite{HP09}. 
Based on this, we extend the definition of concurrent rules with the computation of their application condition. 
In Appendix~\ref{sec:short-cut-rules}, we recall \emph{short-cut rules}~\cite{FKST18}. 
Finally, in Appendix~\ref{properties-adhesive-categories}, we collect some formal results which we built upon in our proofs. 
This mainly encompasses properties of $\M$-adhesive categories. 

\subsection{Nested Conditions}
\label{sec:nested-conditions}
In this section, we recall the definition of \emph{constraints} and \emph{conditions} and of constructions that allow to \enquote{move} conditions along morphisms or rules in such a way that their semantics is preserved. 
Finally, we recall the computation of the application condition of a concurrent rule; this computation assumes those constructions. 

Given an $\M$-adhesive category, \emph{nested constraints} express properties of its objects whereas \emph{nested conditions} express properties of morphisms~\cite{HP09}. 
Conditions are mainly used to restrict the applicability of rules. 
Constraints and conditions are defined recursively as trees of morphisms. 
For the definition of constraints, we assume the existence of an initial object $\emptyset$ in the category $\mathcal{C}$. 

\begin{definition}[(Nested) conditions and constraints]
	Let $\mathcal{C}$ be an $\mathcal{M}$-ad\-he\-sive category with 
	initial object $\emptyset$. 
	Given an object $P$, a \emph{(nested) condition} over $P$ is defined recursively as follows:
	\texttt{true} is a condition over $P$. If $a: P \rightarrow C$ is a 
	morphism and $d$ is a condition over $C$, $\exists \, (a: P \rightarrow C, d)$ is a condition over $P$ again. 
	Moreover, Boolean combinations of conditions over $P$ are conditions over $P$. 
	A \emph{(nested) constraint} is a condition over the initial object $\emptyset$. 
	
	\emph{Satisfaction} of a nested condition $c$ over $P$ for a 
	morphism $g: P \rightarrow G$, denoted as $g \models c$, is defined as follows: 
	Every 
	morphism satisfies \texttt{true}. 
	The morphism $g$ satisfies a condition of the form $c = \exists \, (a: P \rightarrow C, d)$ if there exists an $\mathcal{M}$-morphism $q: C \hookrightarrow G$ such that $g = q \circ a$ and $q \models d$. 
	For Boolean operators, satisfaction is defined as usual. 
	An object $G$ satisfies a constraint $c$, denoted as $G \models c$, if the empty morphism to $G$ does so. 
	A condition $c_1$ over $P$ \emph{implies} a condition $c_2$ over $P$, denoted as $c_1 \Rightarrow c_2$, if for every 
	morphism $g: P \rightarrow G$ with $g \models c_1$ also $g \models c_2$. 
	Two conditions are equivalent, denoted as $c_1 \equiv c_2$, when $c_1 \Rightarrow c_2$ and $c_2 \Rightarrow c_1$.
	Implication and equivalence for constraints is inherited from the respective definition for conditions.
\end{definition}
In the case of graphs, the graph constraints resulting from the above definition are expressively equivalent to a first-order logic on graphs~\cite{HP09}. 

In the following, we recall the \emph{Shift} and \emph{Left} constructions that allow to \enquote{move} constraints along morphisms and rules, respectively. 

\begin{construction}[Shift along morphism]\label{con:shift-morphism}
	Given a condition $c$ over an object $P$ and a morphism $b: P \to P'$, the \emph{shift of $c$ along $b$}, denoted as $\Sh(b,c)$, is inductively defined as follows:
	\begin{itemize}
		\item If $c = \texttt{true}$, 
		\begin{equation*}
			\Sh(b,c) \coloneqq \texttt{true} \enspace .
		\end{equation*}
		\item If $c = \exists \, (a: P \rightarrow C, d)$, 
		\begin{equation*}
			\Sh(b,c) \coloneqq \bigvee_{(a',b') \in \mathcal{F}} \exists \, (a', \Sh(b',d)) \enspace ,
		\end{equation*}
		where $\mathcal{F}$ contains morphism pairs from $\mathcal{E}\rq{}$ with one morphism from $\mathcal{M}$ that complement $a$ and $b$ to a commutative diagram, i.e., 
		\begin{equation*}
			\mathcal{F} \coloneqq \{(a\rq{},b\rq{}) \in \mathcal{E}\rq{} \, | \, b\rq{} \in \mathcal{M} \text{ and } b\rq{} \circ a = a\rq{} \circ b \} \enspace .
		\end{equation*}
		Note that the empty disjunction is equivalent to \texttt{false}.
		
		\begin{figure}
			\centering
			\begin{tikzpicture}
				\matrix (m) [	matrix of math nodes,
											row sep=2em,
											column sep=2em,
											minimum width=2em,
											nodes in empty cells]
				{
						&	P				& \\
					C	&					& P\rq{} \\
						& C\rq{}	& \\};
				\node(d)[left of=m-2-1,node distance=4ex,isosceles triangle,draw=black,fill=lightgray,inner sep=2.5pt,label={[font=\scriptsize,label distance=1pt]180:$d$}]{};
				\node(e)[left of=m-3-2,node distance=4ex,isosceles triangle,draw=black,fill=lightgray,inner sep=2.5pt,label={[font=\scriptsize,label distance=1pt]180:$\Sh(b\rq{},d)$}]{};
				\path[-stealth]
					(m-1-2) edge [->] node [above] {\scriptsize $a$} (m-2-1)
									edge [->] node [above] {\scriptsize $b$} (m-2-3)
					(m-2-1) edge [right hook->] node [below] {\scriptsize $b\rq{}$} (m-3-2)
					(m-2-3) edge [->] node [below] {\scriptsize $a\rq{}$} (m-3-2);
				\end{tikzpicture}
				\caption{Graphical representation of the definition of $\Sh(b, \exists \, (a: P \rightarrow C, d))$.}
				\label{fig:illustration-shift}
			\end{figure}
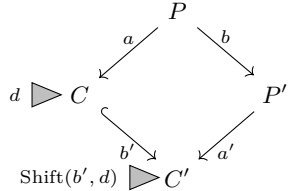
		\item If $c = \neg d$, 
		\begin{equation*}
			\Sh(b,c) \coloneqq \neg \Sh(b,d) \enspace .
		\end{equation*}
		\item If $c = \wedge_{i \in I} c_i$, 
		\begin{equation*}
			\Sh(b,c) \coloneqq \bigwedge_{i \in I} \Sh(b,c_i) \enspace .
		\end{equation*}
	\end{itemize}
\end{construction}

The proof of the correctness of the Shift-construction presupposes an $\mathcal{E}^{\prime}$-$\M$ pair factorization of pairs of morphisms with the same codomain. 
\begin{fact}[{Correctness of $\Sh$~\cite[Lemma~3.11]{EGHLO14}}]\label{fact:correctness-shift-morphism}
	In an $\mathcal{M}$-adhesive category $\mathcal{C}$ with $\mathcal{E}\rq{}$-$\mathcal{M}$ pair factorization, given a nested condition $c$ over an object $P$ and a morphism $b: P \to P\rq{}$, for each morphism $g\rq{}: P\rq{} \to G$
	\begin{equation*}
		g\rq{} \models \Sh(b,c) \Leftrightarrow g \coloneqq g\rq{} \circ b \models c \enspace .
	\end{equation*}
\end{fact}

Similarly, conditions can be \enquote{moved} along rules.
\begin{construction}[Shift over rule]\label{con:shift-rule}
	Given a condition $c$ over an object $R$ and a plain rule $p = (L \xhookleftarrow{l} K \xhookrightarrow{r} R)$, the \emph{shift of $c$ over $p$}, denoted as $\Le(p,c)$, is inductively defined as follows:
	\begin{itemize}
		\item If $c = \texttt{true}$, 
		\begin{equation*}
			\Le(p,c) \coloneqq \texttt{true} \enspace .
		\end{equation*}
		\item If $c = \exists \, (a: R \rightarrow R^*, d)$, consider the following diagram: 
		\begin{center}
			\begin{tikzpicture}
				\matrix (m) [	matrix of math nodes,
											row sep=1em,
											column sep=1em,
											minimum width=1em,
											nodes in empty cells]
				{
					L		&			&	K		&			& R \\
							& (2)	&			&	(1)	&	\\
					L^* &			& K^*	&			& R^* \\};
				\node(sc)[below of=m-3-1,node distance=4ex,isosceles triangle,shape border rotate=90,draw=black,fill=lightgray,inner sep=2.5pt,label={[font=\scriptsize,label distance=1pt]270:$\Le(p^*,d)$}]{$$};
				\node(c)[below of=m-3-5,node distance=4ex,isosceles triangle,shape border rotate=90,draw=black,fill=lightgray,inner sep=2.5pt,label={[font=\scriptsize,label distance=1pt]270:$d$}]{$$};
				\path[-stealth]
					(m-1-3) edge [left hook->] node [above] {\scriptsize $l$} (m-1-1)
									edge [right hook->] node [above] {\scriptsize $r$} (m-1-5)
									edge [->] (m-3-3)
					(m-1-1) edge [->] node [left] {\scriptsize $a^*$} (m-3-1)
					(m-1-5) edge [->] node [right] {\scriptsize $a$} (m-3-5)
					(m-3-3) edge [left hook->] node [below] {\scriptsize $l^*$} (m-3-1)
									edge [right hook->] node [below] {\scriptsize $r^*$} (m-3-5);
			\end{tikzpicture}
		\end{center}
		
		If $a \circ r$ has a pushout complement $(1)$ and 
		\begin{equation*}
			p^* = (L^* \xhookleftarrow{l^*} K^* \xhookrightarrow{r^*} R^*)
		\end{equation*}
		is the rule derived by constructing the pushout $(2)$, i.e., by applying the inverse rule $p^{-1}$ at match $a$ to $R^*$, then 
		\begin{equation*}
			\Le(p,c) \coloneqq \exists \, (a^*: L \to L^*, \Le(p^*,d)) \enspace.
		\end{equation*}
		Otherwise, i.e., if the pushout complement $(1)$ does not exist, 
		\begin{equation*}
			\Le(p,c) \coloneqq \texttt{false} \enspace.
		\end{equation*}
		
		\item If $c = \neg d$, 
		\begin{equation*}
			\Le(p,c) \coloneqq \neg \Le(p,d) \enspace .
		\end{equation*}
		\item If $c = \wedge_{i \in I} c_i$, 
		\begin{equation*}
			\Le(p,c) \coloneqq \bigwedge_{i \in I} \Le(p,c_i) \enspace .
		\end{equation*}
	\end{itemize}
\end{construction}

\begin{fact}[{Correctness of $\Le$~\cite[Lemma~3.14]{EGHLO14}}]\label{fact:correctness-shift-rule}
	In an $\mathcal{M}$-adhesive category $\mathcal{C}$, 
	given a nested condition $c$ over an object $R$ and a plain rule $p = (L \xhookleftarrow{l} K \xhookrightarrow{r} R)$, for each direct transformation $G \Rightarrow_{p,m,m^*} H$, where $m^*$ denotes the co-match of that transformation, 
	\begin{equation*}
		m \models \Le(p,c) \Leftrightarrow m^* \models c \enspace .
	\end{equation*}
\end{fact}

In the presence of application conditions, the definition of a concurrent rule is extended in the following way (for the according diagram, we refer to the main text):
\begin{definition}[$E$-concurrent rule (with application condition)]\label{def:concurrent-rule-with-ac}
	Given two rules $\rho_i = (L_i \xhookleftarrow[]{l_i} K_i \xhookrightarrow[]{r_i} R_i, \mathit{ac}_i)$, where $i = 1,2$, an object $E$ with morphisms $e_1: R_1 \to E$ and $e_2:L_2 \to E$ is an \emph{$E$-dependency relation} for $\rho_1$ and $\rho_2$ if $(e_1,e_2) \in \mathcal{E}^{\prime}$ and the pushout complements $(1a)$ and $(1b)$ for $e_1 \circ r_1$ and $e_2 \circ l_2$ exist.
	
	Given an $E$-dependency relation $E = (e_1,e_2) \in \mathcal{E}^{\prime}$ for rules $\rho_1,\rho_2$, their \emph{$E$-concurrent rule} 
	\begin{equation*}
		\conrule \coloneqq (p, \mathit{ac})
	\end{equation*}
	is defined as
	\begin{equation*}
		p \coloneqq (L \xhookleftarrow[]{l} K \xhookrightarrow[]{r} R)
	\end{equation*}
	where $l \coloneqq l_1^{\prime} \circ k_1$, $r \coloneqq r_2^{\prime} \circ k_2$, $(1a)$, $(1b)$, $(2a)$, and $(2b)$ are pushouts, $(3)$ is a pullback, and 
	\begin{equation*}
		\mathit{ac} \coloneqq \Sh(e_1^{\prime}, \mathit{ac_1}) \wedge \Le(p^{\prime}, \Sh(e_2, \mathit{ac}_2))
	\end{equation*}
	with $p^{\prime} \coloneqq (L \xhookleftarrow{l_1^{\prime}} C_1 \xhookrightarrow{r_1^{\prime}} E)$.
\end{definition}

\subsection{Short-Cut Rules}
\label{sec:short-cut-rules}
We introduced \emph{short-cut rules} as a special kind of sequential rule composition to construct (complex) edit rules from monotonic rules defining a grammar~\cite{FKST18}, and applied them to improve triple graph grammar-based model synchronization processes~\cite{FKST19,FKST20,FKMST20}. 
A \emph{short-cut rule} composes a rule which only deletes (i.e., the inverse rule of a monotonic rule) with a monotonic rule into a single rule whose application has the same effect. 
The knack of the construction is that it allows to identify elements deleted by the first rule as recreated by the second. 
This results in these elements being preserved when the short-cut rule is applied.
We defined the construction of short-cut rules in the context of adhesive categories but restricted to plain monotonic rules and monotonic matching. 
Their construction is based on a \emph{common kernel} for the two given rules. 
We shortly recall the definition of common kernels for monotonic rules in the case of matches restricted to be monomorphisms (which is a special case of Definition~\ref{def:common-kernel_compatibility}). 
There, we understand a common kernel to embed the morphism $k: \Kcap \hookrightarrow V$ into the morphisms $l_1$ and $r_2$. 
However, in the construction of the short-cut rule, the inverse rule $r_1^{-1}$ of the first input rule $r_1: L_1 \hookrightarrow R_1$ is considered (which means that the RHS of the original monotonic rule is its LHS); this is reflected in the following definition. 
In all of the following, we adapt the original notation from~\cite{FKST18} to fit with this paper. 

\begin{definition}[Common kernel for monotonic rules]
	Given two plain, monotonic rules $r_i: L_i \hookrightarrow R_i$, where $i = 1,2$, a \emph{common kernel} for them is a monomorphism $k: \Kcap \hookrightarrow V$ with monomorphisms $u_i: \Kcap \hookrightarrow L_i$ and $v_i: V \hookrightarrow R_i$ such that both induced squares, depicted below, constitute pullback squares. 
	
	\begin{center}
		\begin{tikzpicture}
			\matrix (m) [	matrix of math nodes,
										nodes in empty cells,
										row sep=1em,
										column sep=1em,
										minimum width=1em]
			{
				L_1	&	&	\Kcap		&	& L_2	\\
						&	&					&	&	\\
				R_1	&	&	V	&	& R_2	\\};
			\path[-stealth]
				(m-1-1) edge [right hook->] node [left] {\scriptsize $r_1$} (m-3-1)
				(m-1-3) edge [right hook->] node [left] {\scriptsize $k$} (m-3-3)
								edge [right hook->] node [above] {\scriptsize $u_2$} (m-1-5)
								edge [left hook->] node [above] {\scriptsize $u_1$} (m-1-1)
				(m-3-3) edge [right hook->] node [below] {\scriptsize $v_2$} (m-3-5)
								edge [left hook->] node [below] {\scriptsize $v_1$} (m-3-1)
				(m-1-5) edge [right hook->] node [right] {\scriptsize $r_2$} (m-3-5);
		\end{tikzpicture}
	\end{center}
\end{definition}

Given a common kernel $k: \Kcap \hookrightarrow V$ for monotonic rules $r_1$ and $r_2$, their short-cut rule $\scrule{r_1}{r_2}{k}$ arises by gluing $r_1^{-1}$ and $r_2$ along $k$.
The span $L_1 \xhookleftarrow{u_1} \Kcap \xhookrightarrow{u_2} L_2$ contains the information on how to glue $r_1^{-1}$ and $r_2$ to receive the LHS $L$ and the RHS $R$ of the short-cut rule $\scrule{r_1}{r_2}{k}$.
The morphism $k: \Kcap \hookrightarrow V$ contains the information on how to construct the interface $K$ of the short-cut rule $\scrule{r_1}{r_2}{k}$. 
In case of $\mathbf{Graph}$, for example, $K$ is enhanced by including the elements of $V \setminus \Kcap$, i.e., their difference specifies those elements that would have been deleted by $r_1^{-1}$ and recreated by $r_2$. 

\begin{definition}[Short-cut rule]\label{def:scrule}
	In an adhesive category \textbf{C}, given two monotonic rules $r_i: L_i \hookrightarrow R_i$, where $i = 1,2$, and a common kernel rule $k: \Lcap \hookrightarrow \Rcap$ for them, the \emph{short-cut rule} $\scrule{r_1}{r_2}{k} := (L \xhookleftarrow[]{l} K \xhookrightarrow[]{r} R)$ is computed by executing the following steps:
	\begin{enumerate}
		\item The union $\Lcup$ of $L_1$ and $L_2$ along $\Kcap$ is computed as pushout $(2)$ in Fig.~\ref{fig:constr-LHS-RHS}.
		
		\item The LHS $L$ of the short-cut rule $\scrule{r_1}{r_2}{k}$ is constructed as pushout $(3a)$ in Fig.~\ref{fig:constr-LHS-RHS}.
		
		\item The RHS $R$ of the short-cut rule $\scrule{r_1}{r_2}{k}$ is constructed as pushout $(3b)$ in Fig.~\ref{fig:constr-LHS-RHS}.
					
		\item The interface $K$ of the short-cut rule $\scrule{r_1}{r_2}{k}$ is constructed as pushout $(4)$ in Fig.~\ref{fig:interface-scrule}.
		
		\item Morphisms $l: K \to L$ and $r: K \to R$ are obtained by the universal property of $K$.
	\end{enumerate}
	
	\begin{figure}
		\begin{minipage}[b]{.6\textwidth}
			\centering
			\begin{tikzpicture}
				\matrix (m) [	matrix of math nodes,
											row sep=.3em,
											column sep=.3em,
											minimum width=.3em]
				{
							&				&				&	V				&				&				&	\\
							&				&				&	\phantom{(0)}	&				&				&	\\
							&				&	(1a)	&	\Kcap					&	(1b)	& 			&	\\
							&				&				&	\phantom{(0)}	&				&				&	\\
					R_1 & 			& L_1		& (2)						& L_2		& 			&	R_2 \\
							&	(3a)	& 			& 							&				&	(3b)	& \\
						L	& 			&				& \Lcup					&				&				& R\\};
				\path[-stealth]
					(m-1-4) edge [right hook->,bend left=10] node [right] {\scriptsize $v_2$} (m-5-7)
									edge [left hook->,bend right=10] node [left] {\scriptsize $v_1$} (m-5-1)
					(m-3-4) edge [right hook->] node [right] {\scriptsize $k$} (m-1-4)
									edge [right hook->] node [right] {\scriptsize $u_2$} (m-5-5)
									edge [left hook->] node [left] {\scriptsize $u_1$} (m-5-3)
					(m-5-1) edge [right hook->] node [left] {\scriptsize $e_1^{\prime}$} (m-7-1)
					(m-5-7) edge [right hook->] node [right] {\scriptsize $e_2^{\prime}$} (m-7-7)
					(m-5-3) edge [left hook->] node [above] {\scriptsize $r_1$} (m-5-1)
									edge [right hook->] node [left] {\scriptsize $e_1$} (m-7-4)
					(m-5-5) edge [right hook->] node [above] {\scriptsize $r_2$} (m-5-7)
									edge [right hook->] node [right] {\scriptsize $e_2$} (m-7-4)
					(m-7-4) edge [left hook->] node [below] {\scriptsize $r_1'$} (m-7-1)
									edge [right hook->] node [below] {\scriptsize $r_2'$} (m-7-7);
			\end{tikzpicture}
			\caption{Construction of LHS and RHS of short-cut rule $\scrule{r_1}{r_2}{k}$.}
			\label{fig:constr-LHS-RHS}
		\end{minipage}
		\hfill
		\begin{minipage}[b]{.35\textwidth}
			\begin{center}
				\begin{tikzpicture}
					\matrix (m) [	matrix of math nodes,
												row sep=.3em,
												column sep=.3em,
												minimum width=.3em]
					{
											&	\phantom{(0)}	& L_1	&	\phantom{(0)}	& \\
							\Kcap		&								&	(2)	&								& \Lcup \\
											& 							& L_2	& 							& \\
											& 							& (4)	&								& \\
							V	& 							&			& 							& K \\};
					\path[-stealth]
						(m-2-1) edge [right hook->] node [above] {\scriptsize $u_1$} (m-1-3)
										edge [right hook->] node [below] {\scriptsize $u_2$} (m-3-3)
										edge [dashed,-] (m-2-3)
										edge [right hook->] node [left] {\scriptsize $k$} (m-5-1)
						(m-2-3) edge [->, dashed] (m-2-5)
						(m-5-1) edge [right hook->] node [below] {\scriptsize $z$} (m-5-5)
						(m-2-5) edge [right hook->] node [right] {\scriptsize $\kprime$} (m-5-5)
						(m-1-3) edge [right hook->] node [above] {\scriptsize $e_1$} (m-2-5)
						(m-3-3) edge [right hook->] node [below] {\scriptsize $e_2$} (m-2-5);
				\end{tikzpicture}
				\caption{Construction of interface $K$ of $\scrule{r_1}{r_2}{k}$.}
				\label{fig:interface-scrule}
			\end{center}
		\end{minipage}
	\end{figure}
\end{definition}

\subsection{Properties of $\M$-Adhesive Categories}
\label{properties-adhesive-categories}
Throughout our proofs, we use the following well-known properties of pushouts and pullbacks that hold in any category.
\begin{fact}[Properties of pushouts and pullbacks]\label{fact:properties-pos-pbs}\quad
	\begin{enumerate}[ref={Fact~\ref{fact:properties-pos-pbs}\,(\arabic*)}]
		\item \label{fact:po-comp-decomp} \emph{Pushout composition and decomposition:} Given a commuting diagram like Fig.~\ref{fig:po-pb-decomposition} where $(1)$ is a pushout, $(1)+(2)$ is a pushout if and only if $(2)$ is. 
		\item \label{fact:pb-comp-decomp} \emph{Pullback composition and decomposition:} Given a commuting diagram like Fig.~\ref{fig:po-pb-decomposition} where $(2)$ is a pullback, $(1)+(2)$ is a pullback if and only if $(1)$ is.
		\item \label{fact:po-pb-identity} A pushout along an identity morphism results in an isomorphism; in particular, one can choose this morphism to be an identity morphism as well. Likewise, a pullback along an identity morphism results in an isomorphism and one can choose this morphism to be an identity morphism as well.
		\item \label{fact:pb-mono} Given morphisms $g: A \to B$ and $f: B \to C$ with $f$ being a monomorphism, the span $B \xleftarrow{g} A \xrightarrow{id_A} A$ is a pullback of $(f, f \circ g)$.
	\end{enumerate}
\end{fact}

Also, we frequently exploit the following central properties of $\M$-adhesive categories, also called \emph{HLR properties}. 
These are the properties that make this kind of categories a suitable framework for (double-pushout) rewriting. 
They have first been proven for adhesive categories~\cite{LS05} and also hold for weaker variants like adhesive HLR or $\M$-adhesive categories; see, e.g.,~\cite[Theorem~4.26]{EEPT06} or~\cite[Theorem~4.22]{EEGH15}. 
\begin{fact}[HLR properties of $\M$-adhesive categories]\label{fact:properties-adhesive-categories}
	If $(\CC,\M)$ is an $\M$-adhesive category, the following properties hold:
	\begin{enumerate}
		\item Pushouts along $\M$-morphisms are pullbacks.

		\item If $m$ in Fig.~\ref{fig:po-square} is an $\M$-morphism, pushout complements for $g \circ m$ are unique (up to isomorphism). 
		
		\item $(\CC,\M)$ has $\M$ pushout-pullback decomposition.
		This means that given a diagram like the one depicted in Fig.~\ref{fig:po-pb-decomposition} where the outer square $(1)+(2)$ is a pushout, the right square $(2)$ is a pullback, $w \in \M$, and $l \in \M$ or $k \in \M$, then both $(1)$ and $(2)$ are pushouts and pullbacks. 
	\end{enumerate}
	
	\begin{figure}
		\centering
		\begin{tikzpicture}
			\matrix (m) [	matrix of math nodes,
										nodes in empty cells,
										row sep=1.25em,
										column sep=1.25em,
										minimum width=1.25em,
										nodes={anchor=center}]
			{
				A	&			&	B &			& E \\
					& (1)	&		&	(2)	&	\\
				C &			& D	&			&	F \\};
			\path[-stealth]
				(m-1-3) edge [->] (m-3-3)
								edge [->] node [above] {\scriptsize $r$} (m-1-5)
				(m-1-1) edge [->] node [above] {\scriptsize $k$} (m-1-3)
								edge [->] node [left] {\scriptsize $l$} (m-3-1)
				(m-1-5) edge [->] node [right] {\scriptsize $v$} (m-3-5)
				(m-3-1) edge [->] node [below] {\scriptsize $u$} (m-3-3)
				(m-3-3) edge [->] node [below] {\scriptsize $w$} (m-3-5);
		\end{tikzpicture}
		\caption{Illustration of the $\M$ pushout-pullback and $\M$ pullback-pushout decompositions.}
		\label{fig:po-pb-decomposition}
	\end{figure}
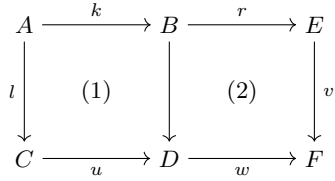
\end{fact}

Moreover, we recall a somewhat more special technical lemma for $\M$-adhesive categories that is needed in one of our proofs. 
It allows to recognize certain squares to be pullbacks. 
\begin{lemma}[{$\M$ pullback-pushout decomposition~\cite[Lemma~B.2]{GHE14}}]\label{lem:M-pb-po-decomp}
	When given a diagram like the one depicted in Fig.~\ref{fig:po-pb-decomposition}, if $(1)+(2)$ is a pullback, $(1)$ is a pushout, $(2)$ commutes, and $v \in \M$, then $(2)$ is a pullback. 
\end{lemma}

Finally, we recall some properties of \emph{initial pushouts} in $\M$-adhesive categories. 
First, initial pushouts can be used to characterize matches for which a rule is applicable.
\begin{fact}[Existence and uniqueness of contexts {\cite[Theorem~6.4]{EEPT06}}]\label{fact:existence-context}
	In an $\M$-adhesive category $(\CC,\M)$ with initial pushouts, given a plain rule $p = (L \xhookleftarrow[]{l} K \xhookrightarrow[]{r} R)$ and a match $m: L \rightarrow G$, the rule $p$ is applicable at match $m$ if and only if there exists a morphism $b_m^*: B_m \to K$ with $l \circ b_m^* = b_m$, where $B_m$ is the boundary object with respect to $m$ and $b_m$ is the boundary over $m$, i.e., where $(1)$ is the initial pushout over $m$ (compare Fig.~\ref{fig:ipo-context-object}).
	
	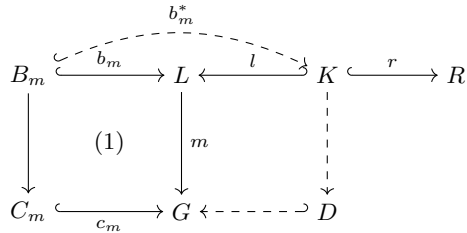
\begin{figure}
		\centering
		\begin{tikzpicture}
			\matrix (m) [	matrix of math nodes,
										nodes in empty cells,
										nodes={anchor=center},
										row sep=1.25em,
										column sep=1.25em,
										minimum width=1.25em]
			{
				B_m	&			&	L	&	& K	&	& R \\
						& (1)	&		& \phantom{(1)}	&		&	& \\
				C_m &			& G	&	&	D	&	& \\};
			\path[-stealth]
				(m-1-3) edge [->] node [right] {\scriptsize $m$} (m-3-3)
				(m-1-1) edge [right hook->] node [above] {\scriptsize $b_m$} (m-1-3)
								edge [right hook->, bend left=25, dashed] node [above] {\scriptsize $b_m^*$} (m-1-5)
								edge [->] (m-3-1)
				(m-1-5) edge [->,dashed] (m-3-5)
								edge [right hook->] node [above] {\scriptsize $r$} (m-1-7)
								edge [left hook->] node [above] {\scriptsize $l$} (m-1-3)
				(m-3-1) edge [right hook->] node [below] {\scriptsize $c_m$} (m-3-3)
				(m-3-5) edge [left hook->,dashed] (m-3-3);
		\end{tikzpicture}
		\caption{Initial pushout and context object.}
		\label{fig:ipo-context-object}
	\end{figure}
\end{fact}

Moreover, initial pushouts enjoy the following closure property with respect to pushouts along $\M$-morphisms.
\begin{fact}[Closure property of initial pushouts {\cite[Lemma~6.5]{EEPT06}}]\label{fact:closure-ipo}
	In an $\M$-adhesive category $(\CC,\M)$ with initial pushouts, given the initial pushout $(1)$ over a morphism $a$ and a pushout $(2)$ along $a$ with $m \in \M$ as depicted in Fig.~\ref{fig:closure-ipo}, the square $(1)+(2)$ constitutes the initial pushout over $d$. 
	
	\begin{figure}
		\centering
		\begin{tikzpicture}
			\matrix (m) [	matrix of math nodes,
										nodes in empty cells,
										row sep=1.25em,
										column sep=1.25em,
										minimum width=1.25em,
										nodes={anchor=center}]
			{
				B_a	&			&	A						&			& D \\
						& (1)	&							&	(2)	&	\\
				C_a &			& A^{\prime}	& 		&	D^{\prime} \\};
			\path[-stealth]
				(m-1-3) edge [->] node [right] {\scriptsize $a$} (m-3-3)
								edge [right hook->] node [above] {\scriptsize $m$} (m-1-5)
				(m-1-1) edge [right hook->] node [above] {\scriptsize $b_a$} (m-1-3)
								edge [->] (m-3-1)
				(m-1-5) edge [->] node [right] {\scriptsize $d$} (m-3-5)
				(m-3-1) edge [right hook->] node [below] {\scriptsize $c_a$} (m-3-3)
				(m-3-3) edge [right hook->] node [below] {\scriptsize $n$} (m-3-5);
		\end{tikzpicture}
		\caption{First closure property of initial pushouts.}
		\label{fig:closure-ipo}
	\end{figure}
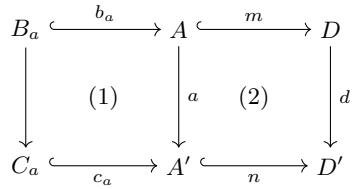
	
	Moreover, given the initial pushout $(1)$ over a morphism $a$ and a pushout $(2)$ along $d$ with $m \in \M$ as depicted in Fig.~\ref{fig:closure-ipo-2}, the square $(3)$ constitutes the initial pushout over $d$, where $b_a^{\prime},c_a^{\prime}$ are the morphisms induced by initiality of $(1)$. 
	
	\begin{figure}
		\centering
		\begin{tikzpicture}
			\matrix (m) [	matrix of math nodes,
										nodes in empty cells,
										row sep=1.25em,
										column sep=1.25em,
										minimum width=1.25em,
										nodes={anchor=center}]
			{
				B_a	&			&	A						&			& D 					&	& & B_a	&			&	D \\
						& (1)	&							&	(2)	&							&	&	&			& (3)	& \\	
				C_a &			& A^{\prime}	& 		&	D^{\prime}	&	& &	C_a	&			& D^{\prime} \\};
			\path[-stealth]
				(m-1-3) edge [->] node [right] {\scriptsize $a$} (m-3-3)
				(m-1-5) edge [left hook->] node [above] {\scriptsize $m$} (m-1-3)
				(m-1-1) edge [right hook->] node [above] {\scriptsize $b_a$} (m-1-3)
								edge [->] node [left] {\scriptsize $x_a$} (m-3-1)
								edge [right hook->,bend left,dashed] node [above] {\scriptsize $b_a^{\prime}$} (m-1-5)
				(m-1-5) edge [->] node [right] {\scriptsize $d$} (m-3-5)
				(m-3-1) edge [right hook->] node [below] {\scriptsize $c_a$} (m-3-3)
								edge [right hook->,bend right,dashed] node [below] {\scriptsize $c_a^{\prime}$} (m-3-5)
				(m-3-5) edge [left hook->] node [below] {\scriptsize $n$} (m-3-3)
				(m-1-8) edge [right hook->] node [above] {\scriptsize $b_a^{\prime}$} (m-1-10)
								edge [->] node [left] {\scriptsize $x_a$} (m-3-8)
				(m-3-8) edge [right hook->] node [below] {\scriptsize $c_a^{\prime}$} (m-3-10)
				(m-1-10) edge [->] node [right] {\scriptsize $d$} (m-3-10);
		\end{tikzpicture}
		\caption{Second closure property of initial pushouts.}
		\label{fig:closure-ipo-2}
	\end{figure}
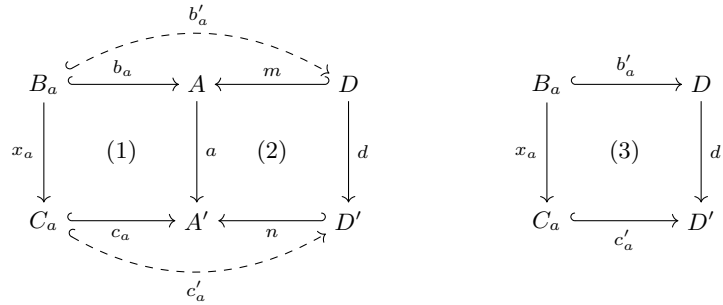
\end{fact}

	\section{Technical Lemma and Proofs}
	\label{sec:technical-appendix}
	In this section, we present the proofs of all our results and an additional (technical) lemma needed for on of them. 

We start with this general lemma.
It is concerned with the composition of an initial pushout with a pullback in $\M$-adhesive categories and seems to be new. 

\begin{lemma}[Interaction of initial pushouts with pullbacks]\label{lem:interaction-ipos-pbs}
	In an $\M$-adhesive category $(\CC,\M)$, given a diagram like the one in Fig.~\ref{fig:interaction-ipos-pbs} where $(1)$ is an initial pushout, $(2)$ a pullback with $a_2\in \M$, and $(3)$ a pushout such that $a_3 \in \M$, there are unique $\M$-morphisms $d: B_{f_1} \hookrightarrow A_5$ and $e: C_{f_1} \hookrightarrow A_6$ such that $a_3 \circ d = a_1 \circ b_{f_1}$, $a_4 \circ e = a_2 \circ c_{f_1}$. 
	Moreover, the thereby induced square $B_{f_1} \xrightarrow{d} A_5 \xrightarrow{f_3} A_6 \xleftarrow{e} C_{f_1} \xleftarrow{x_{f_1}} B_{f_1}$ is a pullback.
	
	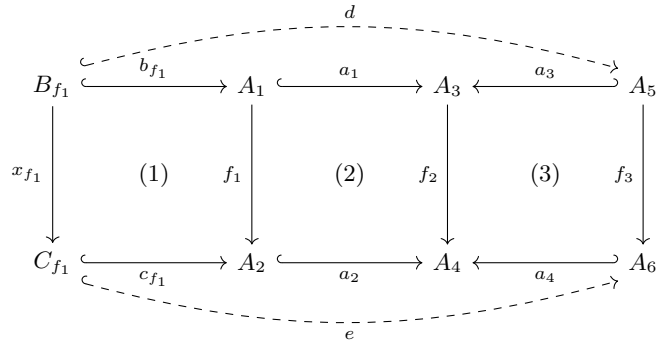
\begin{figure}
		\centering
		\begin{tikzpicture}
			\matrix (m) [	matrix of math nodes,
										nodes in empty cells,
										nodes={anchor=center},
										row sep=2em,
										column sep=2em,
										minimum width=2em]
			{	B_{f_1}	&			&	A_1	&			& A_3	&			&	A_5 \\
								& (1)	&			& (2)	&			& (3)	& \\
				C_{f_1}	&			& A_2	&			& A_4	&			& A_6 \\};
			\path[-stealth]
				(m-1-1) 						edge [right hook->] node [above] {\scriptsize $b_{f_1}$} (m-1-3)
														edge [->] node [left] {\scriptsize $x_{f_1}$} (m-3-1)
				(m-1-1.north east)	edge [right hook->,bend left=15,dashed] node [above] {\scriptsize $d$} (m-1-7.north west)
				(m-1-3) 						edge [right hook->] node [above] {\scriptsize $a_1$} (m-1-5)
														edge [->] node [left] {\scriptsize $f_1$} (m-3-3)
				(m-1-5) 						edge [->] node [left] {\scriptsize $f_2$} (m-3-5)
				(m-1-7) 						edge [left hook->] node [above] {\scriptsize $a_3$} (m-1-5)
														edge [->] node [left] {\scriptsize $f_3$} (m-3-7)
				(m-3-1) 						edge [right hook->] node [below] {\scriptsize $c_{f_1}$} (m-3-3)
				(m-3-1.south east)	edge [right hook->,bend right=15,dashed] node [below] {\scriptsize $e$} (m-3-7.south west)
				(m-3-3) 						edge [right hook->] node [below] {\scriptsize $a_2$} (m-3-5)
				(m-3-7) 						edge [left hook->] node [below] {\scriptsize $a_4$} (m-3-5);
		\end{tikzpicture}
		\caption{Interaction of initial pushouts with pullbacks.}
		\label{fig:interaction-ipos-pbs}
	\end{figure}
\end{lemma}
	
\begin{proof}
	First, compute the cube that is depicted in Fig.~\ref{fig:proof-interaction-ipos-pbs} by computing the missing side faces as pullbacks: 
	The pullbacks in the front and back exist since $a_4,a_3 \in \M$; the resulting morphisms $d,e,x,y$ are $\M$-morphisms as they result from pullbacks along $\M$-morphisms. 
	The morphism $\bar{x_{f_1}}: Y \to X$ such that the whole cube commutes is then induced by the universal property of $X$ as pullback object; moreover, by pullback decomposition, the induced right side face is a pullback, as well. 
	Since $(\CC,\M)$ is $\M$-adhesive and $a_3 \in \M$, the pushout at the bottom of the cube has the weak vertical van Kampen property. 
	Since $a_2 \circ c_{f_1}, a_1 \circ b_{f_1},e,d \in \M$, this implies that the top face is a pushout as well.
	(If $(\CC,\M)$ is even adhesive HLR and we drop the assumption $a_2$ (and hence $a_2 \circ c_{f_1}$) $\in \M$, we cannot conclude $e,d \in \M$. However, the top square still is a pushout because the bottom pushout then has the van Kampen property.)
	We prove the statement by proving that, without loss of generality, $x$ and $y$ are the identity morphisms of $C_{f_1}$ and $B_{f_1}$, respectively. 
	
	\begin{figure}
		\centering
		\begin{tikzpicture}
			\matrix (m) [	matrix of math nodes,
										nodes in empty cells,
										nodes={anchor=center},
										row sep=2em,
										column sep=2em,
										minimum width=2em]
			{	
									& B_{f_1}	& 		& Y \\
					C_{f_1}	&					&	X		& \\
									&	A_1			&			& \\
				 A_2			&					&		 	& \\
									& A_3			& 		& A_5 \\
				 A_4			&					& A_6	& \\};
			\node (1) at ($(m-1-2)!0.5!(m-4-1)$) {$(1)$};
			\node (2) at ($(m-3-2)!0.5!(m-6-1)$) {$(2)$};
			\node (3) at ($(m-6-1)!0.5!(m-5-4)$) {$(3)$};
			\path[-stealth]
				(m-1-2) edge [->,shorten >= -1pt] node [left,inner xsep=5pt,pos=.25] {\scriptsize $x_{f_1}$} (m-2-1)
								edge [right hook->] node [right,near end] {\scriptsize $b_{f_1}$} (m-3-2)
				(m-1-2.north) edge [right hook->,dashed,bend left=25] node [above,pos=.7] {\scriptsize $b_{f_1}^*$} (m-1-4.north)
				(m-1-4) edge [left hook->,dashed] node [above,pos=.35] {\scriptsize $y$} (m-1-2)
								edge [right hook->,dashed] node [right] {\scriptsize $d$} (m-5-4)
								edge [->,dashed] node [right,inner xsep=5pt] {\scriptsize $\bar{x_{f_1}}$} (m-2-3)
				(m-2-1) edge [right hook->] node [left] {\scriptsize $c_{f_1}$} (m-4-1)
				(m-3-2) edge [->] node [above] {\scriptsize $f_1$} (m-4-1)
								edge [right hook->] node [right] {\scriptsize $a_1$} (m-5-2)
				(m-4-1) edge [right hook->] node [left] {\scriptsize $a_2$} (m-6-1)
				(m-5-2) edge [->] node [above] {\scriptsize $f_2$} (m-6-1)
				(m-5-4) edge [left hook->] node [above,near start] {\scriptsize $a_3$} (m-5-2)
								edge [->] node [right] {\scriptsize $f_3$} (m-6-3)
				(m-6-3) edge [left hook->] node [above] {\scriptsize $a_4$} (m-6-1)
				(m-2-3) edge [-,draw=white, line width=4pt] (m-2-1)
								edge [left hook->,dashed] node [above,pos=.35] {\scriptsize $x$} (m-2-1)
								edge [-,draw=white,line width=4pt] (m-6-3)
								edge [right hook->,dashed] node [right] {\scriptsize $e$} (m-6-3)
				(m-2-1.north) edge [-,draw=white,line width=3pt,bend left=25] (m-2-3.north)
											edge [right hook->,dashed,bend left=25] node [above,pos=.7] {\scriptsize $c_{f_1}^*$} (m-2-3.north);
		\end{tikzpicture}
		\caption{Proving the interaction of initial pushouts with pullbacks.}
		\label{fig:proof-interaction-ipos-pbs}
	\end{figure}
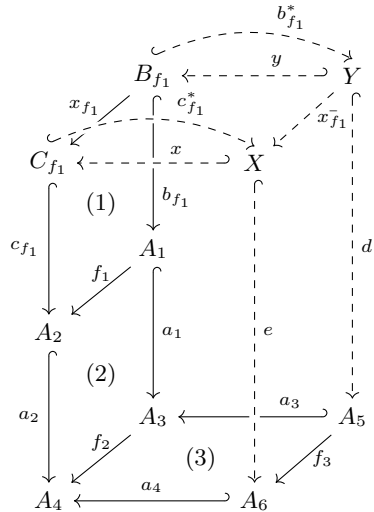
	
	Composing the top pushout with the initial pushout $(1)$ results in a second pushout over $f_1: A_1 \to A_2$. 
	Since $x,y \in \M$ also $b_{f_1} \circ y, c_{f_1} \circ x \in \M$. 
	Then, by initiality of $(1)$, we obtain $\M$-morphisms $b_{f_1}^{*}: B_{f_1} \hookrightarrow Y$ and $c_{f_1}^{*}: C_{f_1} \hookrightarrow X$ such that 
	\begin{equation*}
		b_{f_1} \circ y \circ b_{f_1}^{*} = b_{f_1} \text{ and } c_{f_1} \circ x \circ c_{f_1}^{*} = c_{f_1} \enspace ; 
	\end{equation*}
	compare the top of Fig.~\ref{fig:proof-interaction-ipos-pbs}. 
	Canceling the monomorphisms $b_{f_1}$ and $c_{f_1}$, respectively, shows that $x$ and $y$ are split epi, and hence isomorphisms. 
	Thus, without loss of generality, $X = C_{f_1}$, $Y = B_{f_1}$, $x = id_{C_{f_1}}$, $y = id_{B_{f_1}}$, and $\bar{x_{f_1}} =  x_{f_1}$. 
	In particular, $a_3 \circ d = a_1 \circ b_{f_1}$, $a_4 \circ e = a_2 \circ c_{f_1}$, and the desired square is a pullback. \qed
\end{proof}

The rest of this section contains the proofs of all results presented in the main text of the paper. 

\begin{proof}[of Lemma~\ref{lem:existence-extension-morphism}]
	By compatibility of the common kernel $k$ with $E$ and the definition of a concurrent rule we have that 
	\begin{align*}
		r_1^{\prime} \circ e_1^{\prime\prime} \circ u_1	& = e_1 \circ r_1 \circ u_1 \\
																										& = e_2 \circ l_2 \circ u_2 \\
																										& = l_2^{\prime} \circ e_2^{\prime\prime} \circ u_2 \enspace .
	\end{align*}
	Thus, by the universal property of the pullback computing $K$, we obtain a unique morphism $p: \Kcap \to K$ such that $k_i \circ p = e_i^{\prime\prime} \circ u_i$ for $i=1,2$. 
	Moreover, $p \in \M$ by decomposition of $\M$-morphisms. \qed
\end{proof}

\begin{proof}[of Proposition~\ref{prop:con-rule-as-gen-con-rule}]
	First, since pushouts and pullbacks along isomorphisms result in isomorphisms again, and since in $\M$-adhesive categories pushouts along $\M$-morphisms are pullbacks, $\kprime$ is an isomorphism if and only if $k$ is one. 
	
	Furthermore, since in $\M$-adhesive categories pullbacks along $\M$-morphisms exist, for every $E$-dependency relation $E = (e_1,e_2)$ we can construct $id_{\Kcap}$ as compatible common kernel as follows: 
	We obtain $u_i: \Kcap \hookrightarrow K_i$, where $i=1,2$, by pulling back the pair of $\M$-morphisms $(e_1 \circ r_1, e_2 \circ l_2)$. 
	Thus, it suffices to ensure that there are suitable morphisms $v_1: V = \Kcap \to L_1, v_2: V = \Kcap \to R_2$ such that $k$ indeed constitutes a common kernel for $\rho_1$ and $\rho_2$ compatible with $E$. 
	This is guaranteed by \ref{fact:pb-mono} when setting $v_1 \coloneqq l_1 \circ u_1: V = \Kcap \to L_1, v_2 \coloneqq r_2 \circ u_2: V = \Kcap \to R_2$ since $l_1$ and $r_2$ are monic (as $\M$-morphisms). 
	
	Finally, without loss of generality, we obtain $\kprime = k: K \hookrightarrow K$ and $\pprime = p$. 
	In particular, $\lprime = l_1^{\prime} \circ k_1 = l$ and $\rprime = r_2^{\prime} \circ k_2 = r$. \qed
\end{proof}

\begin{proof}[of Proposition~\ref{prop:sc-rule-is-gcr}]
	To make the constructions comparable, we consider the (equivalent) rules $\rho_1^{-1} = (R_1 \xhookleftarrow{r_1} L_1 \xhookrightarrow{id_{L_1}} L_1)$ and $\rho_2 = (L_2 \xhookleftarrow{id_{L_2}} L_2 \xhookrightarrow{r_2} R_2)$, instead. 
	
	\begin{figure}
		\centering
		\includegraphics[width=\textwidth]{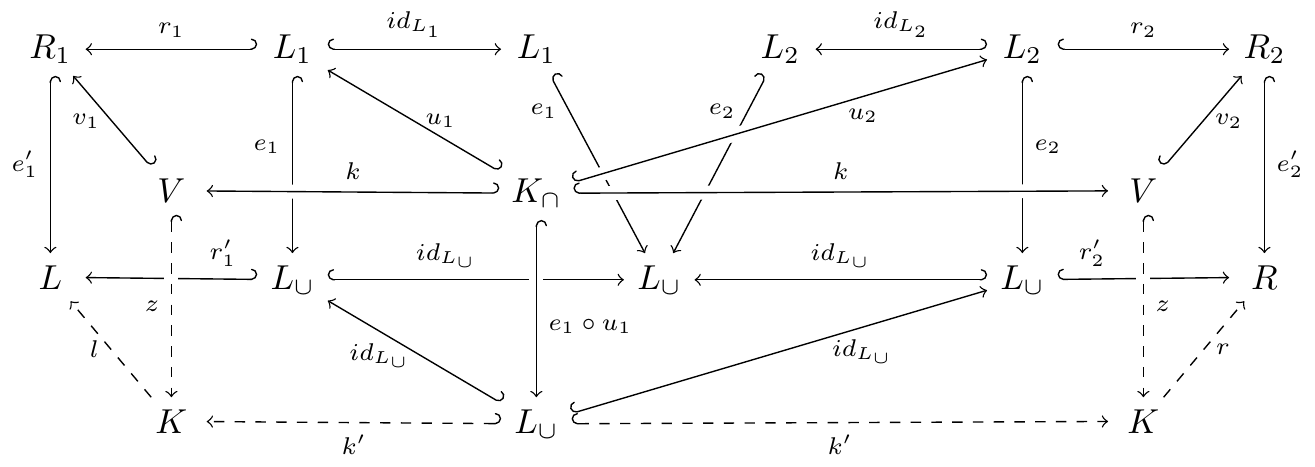}
		\caption{Short-cut rule as generalized concurrent rule.}
		\label{fig:sc-rule-as-gen-con-rule}
	\end{figure}
	
	Applying the construction of GCRs to $\rho_1^{-1}$ and $\rho_2$ with $E = (e_1: L_1 \hookrightarrow \Lcup, e_2: L_2 \hookrightarrow \Lcup)$ as $E$-dependency relation and $k: \Kcap \hookrightarrow V$ as common kernel, results in the diagram depicted in Fig.~\ref{fig:sc-rule-as-gen-con-rule} (where we employ the notation of Figs.~\ref{fig:constr-LHS-RHS} and \ref{fig:interface-scrule}): 
	In this special case, the morphism $p$ provided by Lemma~\ref{lem:existence-extension-morphism} is already the morphism $e_1 \circ u_1 = e_2 \circ u_2$. 
	Therefore, it is evident that this computes the same rule as the short-cut rule construction, i.e., $\rho_1^{-1} *_{E,k} \rho_2 = \scrule{\rho_1}{\rho_2}{k}$. 
	In particular, as pushout complements for a sequence of two morphisms with first morphism an identity always exist, $E$ indeed is an $E$-dependency relation, as long as $(e_1,e_2) \in \mathcal{E}^{\prime}$. 
	Since this pair is computed as pushout along the pair of $\M$-morphisms $(u_1,u_2)$, $E$ is a pair of jointly epic $\M$-morphisms; a class of morphisms that is regularly included in (or even constitutes) $\mathcal{E}^{\prime}$ in practical applications. \qed
\end{proof}

\begin{proof}[of Proposition~\ref{prop:embedding-characterization-monic}]
	First, if the application of Construction~\ref{con:span-gen-con-rule-1} results in a GCR, i.e., if $\lprime,\rprime \in \M$, we immediately obtain $\lprime \circ \pprime = e_1^{\prime} \circ v_1 \in \M$ by composition and then $v_1 \in \M$ by decomposition of $\M$-morphisms. 
	Analogously, $\rprime \in \M$ implies $v_2 \in \M$. 
	
	\begin{figure}
		\centering
		\begin{tikzpicture}
			\matrix (m) [	matrix of math nodes,
										row sep=2em,
										column sep=2em,
										minimum width=2em,
										nodes in empty cells]
			{
							&	K				&					&	& \\
				\Kcap	&					&	\Kprime	&	&	L \\
							& V	&					&	&	\\};
			\path[-stealth]
				(m-2-1) edge [right hook->] node [above] {\scriptsize $p$} (m-1-2)
								edge [right hook->] node [below] {\scriptsize $k$} (m-3-2)
				(m-1-2) edge [right hook->] node [above] {\scriptsize $\kprime$} (m-2-3)
								edge [right hook->,bend left=15] node [above] {\scriptsize $l_1^{\prime} \circ k_1$} (m-2-5)
				(m-3-2) edge [right hook->] node [below] {\scriptsize $\pprime$} (m-2-3)
								edge [right hook->,bend right=15] node [below] {\scriptsize $e_1^{\prime} \circ v_1$} (m-2-5)
				(m-2-3) edge [->,dashed] node [above] {\scriptsize $\lprime$} (m-2-5);
		\end{tikzpicture}
		\caption{Obtaining $\lprime\in\M$ via $\M$-effective unions.}
		\label{fig:obtaining-lprime-in-M}
	\end{figure}
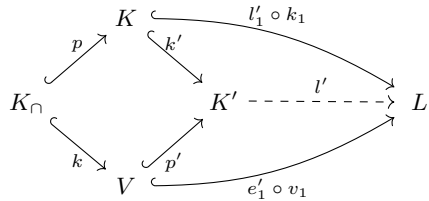
	
	For the converse direction, first Fig.~\ref{fig:obtaining-lprime-in-M} illustrates how $\lprime$ is obtained by the universal property of the pushout computing $\Kprime$. 
	In particular, $v_1 \in \M$ implies $e_1^{\prime} \circ v_1 \in \M$ by composition of $\M$-morphisms. 
	Thus, if the outer square is a pullback, unions being $\M$-effective implies that $\lprime \in \M$. 
	To show the outer square to be a pullback, compare Fig.~\ref{fig:proving-outer-square-pb}: 
	The two top squares are pullbacks by assumption (the top square being a pushout along an $\M$-morphism) and the bottom square is a pullback according to \ref{fact:pb-mono}. 
	Then, pullback composition (\ref{fact:pb-comp-decomp}) implies that the whole square is a pullback, indeed. 
	
	\begin{figure}
		\centering
		\begin{tikzpicture}
			\matrix (m) [	matrix of math nodes,
										row sep=1.25em,
										column sep=1.25em,
										minimum width=1.25em,
										nodes in empty cells]
			{
					&			& L			&				& \\
					& C_1	&				& L_1		&	\\
				K	&			&	K_1		&				& V \\
					&			&				&	\Kcap	&	\\
					&			& \Kcap	&				&	\\};
			\path[-stealth]
				(m-5-3) edge [left hook->] node [left,inner xsep=5pt] {\scriptsize $p$} (m-3-1)
								edge [right hook->] node [right,inner xsep=5pt,pos=.3] {\scriptsize $id_{\Kcap}$} (m-4-4)
				(m-4-4) edge [left hook->] node [left,pos=.3] {\scriptsize $u_1$} (m-3-3)
								edge [right hook->] node [right,inner xsep=5pt,pos=.3] {\scriptsize $k$} (m-3-5)
				(m-3-1) edge [right hook->] node [above,inner ysep=3pt,pos=.4] {\scriptsize $k_1$} (m-2-2)
				(m-3-3) edge [left hook->] node [left,pos=.3,inner sep=5pt] {\scriptsize $e_1^{\prime\prime}$} (m-2-2)
								edge [right hook->] node [right,inner xsep=5pt,pos=.3] {\scriptsize $l_1$} (m-2-4)
				(m-3-5) edge [left hook->] node [right,pos=.7,inner xsep=5pt] {\scriptsize $v_1$} (m-2-4)
				(m-2-2) edge [right hook->] node [above,inner ysep=3pt,pos=.4] {\scriptsize $l_1^{\prime}$} (m-1-3)
				(m-2-4) edge [left hook->] node [right,pos=.7,inner xsep=7pt] {\scriptsize $e_1^{\prime}$} (m-1-3);
		\end{tikzpicture}
		\caption{Proving the outer square of Fig.~\ref{fig:obtaining-lprime-in-M} to be a pullback.}
		\label{fig:proving-outer-square-pb}
	\end{figure}
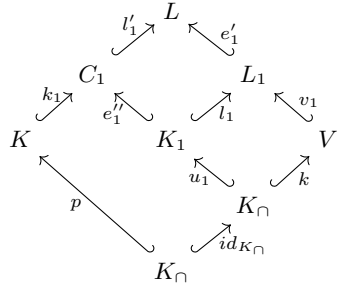
	
	Finally, we show $\M$-effective unions to be necessary for this result to hold by constructing an (abstract) counterexample otherwise.  
	Whenever a category $(\CC,\M)$ does not have $\M$-effective unions, there is a pullback of $\M$-morphisms $l_1,v_1$ witnessing this; in particular $\lprime \notin \M$ (as depicted in Fig.~\ref{fig:counterexample-M-effective-unions}). 
	\begin{figure}
		\centering
		\begin{tikzpicture}
			\matrix (m) [	matrix of math nodes,
										row sep=2em,
										column sep=2em,
										minimum width=2em,
										nodes in empty cells,
										nodes={anchor=center}]
			{
							&	K_1			&					&	& \\
				\Kcap	&					&	\Kprime	&	&	L_1 \\
							& V	&					&	&	\\};
			\path[-stealth]
				(m-2-1) edge [right hook->] node [above] {\scriptsize $u_1$} (m-1-2)
								edge [right hook->] node [below] {\scriptsize $k$} (m-3-2)
				(m-1-2) edge [right hook->] node [above] {\scriptsize $\kprime$} (m-2-3)
								edge [right hook->,bend left=15] node [above] {\scriptsize $l_1$} (m-2-5)
				(m-3-2) edge [right hook->] node [below] {\scriptsize $\pprime$} (m-2-3)
								edge [right hook->,bend right=15] node [below] {\scriptsize $v_1$} (m-2-5)
				(m-2-3) edge [->,dashed] node [above] {\scriptsize $\lprime$} (m-2-5);
		\end{tikzpicture}
		\caption{Counterexample to $\M$-effective unions.}
		\label{fig:counterexample-M-effective-unions}
	\end{figure}
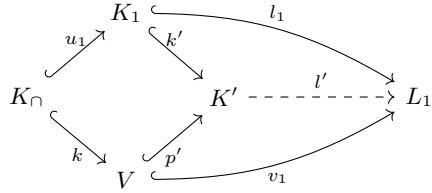
	Choose $\Eprime$ such that it includes the class of pairs of jointly epimorphic $\M$-morphisms. 
	Let $\rho_1 = (L_1 \xhookleftarrow{l_1} K_1 \xhookrightarrow{id_{K_1}} K_1)$, $\rho_2 = (\Kcap \xhookleftarrow{id_{\Kcap}} \Kcap \xhookrightarrow{k} V)$, and $E = (id_{K_1},u_1) \in \Eprime$. 
	Embed $k$ via $u_1,v_1$ into the left part of $\rho_1$ and via $u_2 = id_{\Kcap}$ and $v_2 = id_{V}$ into the right part of $\rho_2$.
	This makes $k$ into a common kernel compatible with $E$. 
	However, applying Construction~\ref{con:span-gen-con-rule-1} results in $L_1 \xleftarrow{\lprime} \Kprime \xhookrightarrow{id_{\Kprime}} \Kprime$ which is not a rule since $\lprime \notin \M$ by assumption. 
	The diagram in Fig.~\ref{fig:abstract-counterexample-M-effective-unions} depicts the detailed computation of this abstract counterexample. 
	\begin{figure}
		\centering
		\includegraphics[width=\textwidth]{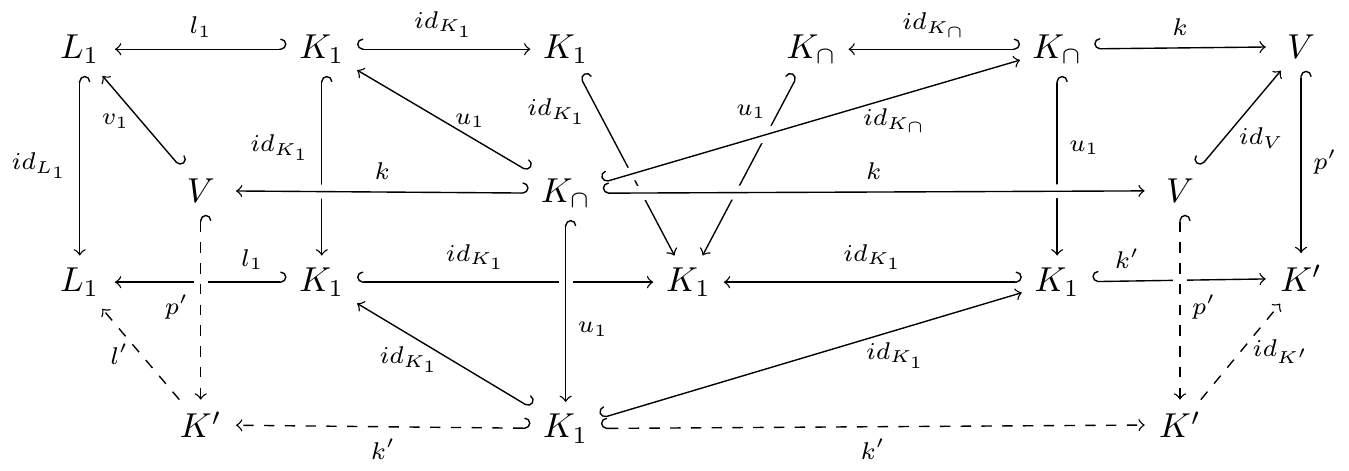}
		\caption{Abstract counterexample to Proposition~\ref{prop:embedding-characterization-monic} in absence of $\M$-effective unions.}
		\label{fig:abstract-counterexample-M-effective-unions}
	\end{figure}
	By the prevalence of identities, it is not difficult to check the occurring squares to be pushouts resp. pullbacks as needed. 
	Moreover, $(id_{K_1},u_1) \in \Eprime$: both morphisms are clearly $\M$-morphisms and the pair is jointly epimorphic since $id_{K_1}$ is even an epi. \qed
\end{proof}

\begin{proof}[of Proposition~\ref{prop:enhancement-characterization-monic}]
	First, let $L \xhookleftarrow{\lprime} \Kprime \xhookrightarrow{\rprime} R$ be a GCR, i.e., assume that there exists a common kernel $k$ for $\rho_1$ and $\rho_2$ that is compatible with $E$ such that $\genconrule = L \xhookleftarrow{\lprime} \Kprime \xhookrightarrow{\rprime} R$. 
	Compare Fig.~\ref{fig:proof-appropriate-enhancement} for the following.
	
	\begin{figure}
		\centering
		\includegraphics{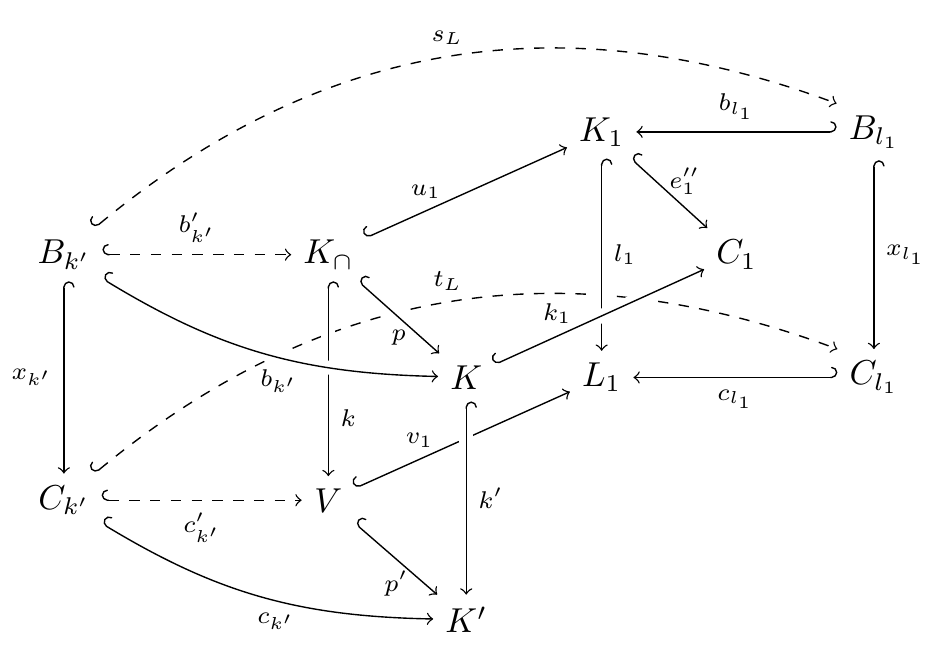}
		\caption{Proving appropriate enhancement.}
		\label{fig:proof-appropriate-enhancement}
	\end{figure}
	
	The solid squares show the relevant part of the computation of $\genconrule$ (namely, the embedding of the common kernel $k$ into $l_1$, the computation of $\Kprime$ as a pushout, and a part of the computation of $K$ via pullback -- the morphism $k_1$) and the initial pushouts over $\kprime$ (the bent square at the front) and $l_1$ (the square to the very right). 
	First, Fact~\ref{fact:closure-ipo} ensures that the boundary and context objects $B_{\kprime}$ and $C_{\kprime}$ of the initial pushout over $\kprime$ also constitute the boundary and context objects of the initial pushout over $k$ (as $\kprime$ is computed as pushout along $k$), where the necessary morphisms $b_{\kprime}^{\prime}, c_{\kprime}^{\prime}$ are induced by initiality. 
	Then, the sequence of three squares at the center of the figure is as in the situation of Lemma~\ref{lem:interaction-ipos-pbs}: the first square is an initial pushout followed by a pullback and the opposing square is a pushout. 
	Hence, Lemma~\ref{lem:interaction-ipos-pbs} implies the existence of $s_L$ and $t_L$ yielding the required pullback square (bent square in the background). 
	
	Finally, using the commutativity of the whole diagram we compute
	\begin{align*}
		k_1 \circ b_{\kprime}	& = k_1 \circ p \circ b_{\kprime}^{\prime} \\
													& = e_1^{\prime\prime} \circ u_1 \circ b_{\kprime}^{\prime} \\
													& = e_1^{\prime\prime} \circ b_{l_1} \circ s_L
	\end{align*}
	as was to be shown. 
	The existence of $s_R: B_{\kprime} \hookrightarrow B_{r_2}$ and $t_R:C_{\kprime} \hookrightarrow C_{r_2}$ such that the induced square is a pullback and $k_2 \circ b_{\kprime} = e_2^{\prime\prime} \circ b_{r_2} \circ s_R$ is shown completely analogously. 
	
	For the other direction, assume $\kprime$ to be appropriately enhancing. 
	Using the equations $k_1 \circ b_{\kprime} = e_1^{\prime\prime} \circ b_{l_1} \circ s_L$ and $k_2 \circ b_{\kprime} = e_2^{\prime\prime} \circ b_{r_2} \circ s_R$, we first compute 
	\begin{align*}
		r_1^{\prime} \circ e_1^{\prime\prime} \circ b_{l_1} \circ s_L & = r_1^{\prime} \circ k_1 \circ b_{\kprime} \\
																																	& = l_2^{\prime} \circ k_2 \circ b_{\kprime}\tag{1}\label{al:appr-enhancing} \\
																																	& = l_2^{\prime} \circ e_2^{\prime\prime} \circ b_{r_2} \circ s_R \enspace. 
	\end{align*}
	We then compute $(u_1:\Kcap \hookrightarrow K_1,u_2:\Kcap \hookrightarrow K_2)$ as pullback of $(e_1 \circ r_1 = r_1^{\prime} \circ e_1^{\prime\prime},e_2 \circ l_2 = l_2^{\prime} \circ e_2^{\prime\prime})$. 
	Then, the universal property of this pullback and Eq.~\ref{al:appr-enhancing} imply the existence of a unique morphism $b_{\kprime}^{\prime}: B_{\kprime} \to \Kcap$ such that $u_1 \circ b_{\kprime}^{\prime} = b_{l_1} \circ s_L$ and $u_2 \circ b_{\kprime}^{\prime} = b_{r_2} \circ s_R$; moreover, $b_{\kprime}^{\prime} \in \M$ by decomposition of $\M$-morphisms. 
	We then compute the object $V$ as pushout of $x_{\kprime}$ along this morphism $b_{\kprime}^{\prime}$; this results in Fig.~\ref{fig:obtaining-common-kernel}: 
	The left square is the computed pushout, resulting in the $\M$-morphism $k: \Kcap \hookrightarrow V$. 
	The outer square is a pullback, which exists by assumption (composing the assumed pullback with the initial pushout). 
	The morphism $v_1$ is obtained by the universal property of the pushout and makes the whole diagram commute. 
	In particular, as the diagram commutes, $l_1 \in \M$, the outer square is a pullback, and the left square a pushout, $\M$-pullback-pushout decomposition is applicable and ensures the right square to constitute a pullback. 
	Completely analogously, one constructs the pullback embedding $k$ into $r_2$. 
	
	\begin{figure}
		\centering
		\begin{tikzpicture}
			\matrix (m) [	matrix of math nodes,
										nodes in empty cells,
										row sep=1.5em,
										column sep=1.5em,
										minimum width=1.5em]
			{
				B_{\kprime}	&	&	\Kcap		&	& K_1 \\
										&	&					&	&	\\
				C_{\kprime} &	& V	&	&	L_1 \\};
			\path[-stealth]
				(m-1-3) edge [right hook->,dashed] node [right] {\scriptsize $k$} (m-3-3)
								edge [right hook->] node [above] {\scriptsize $u_1$} (m-1-5)
				(m-1-1) edge [right hook->] node [above] {\scriptsize $b_{\kprime}^{\prime}$} (m-1-3)
								edge [right hook->] node [left] {\scriptsize $x_{\kprime}$} (m-3-1)
								edge [right hook->,bend left=28] node [above] {\scriptsize $b_{l_1} \circ s_L$} (m-1-5)
				(m-1-5) edge [right hook->] node [right] {\scriptsize $l_1$} (m-3-5)
				(m-3-1) edge [right hook->,dashed] node [below] {\scriptsize $c_{\kprime}^{\prime}$} (m-3-3)
								edge [right hook->,bend right=28] node [below] {\scriptsize $c_{l_1} \circ t_L$} (m-3-5)
				(m-3-3) edge [right hook->,dashed] node [below] {\scriptsize $v_1$} (m-3-5);
		\end{tikzpicture}
		\caption{Obtaining the common kernel $k: \Kcap \hookrightarrow V$ via pushout.}
		\label{fig:obtaining-common-kernel}
	\end{figure}
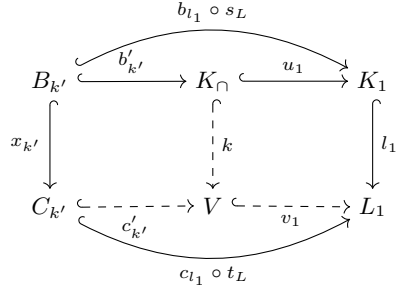

	With almost the same argument we ensure that $k$ computes the correct interface $\Kprime$ (see Fig.~\ref{fig:ensuring-correct-interface}):
	The left square, again, is the pushout computing $V$ and the outer square the given initial pushout over $\kprime$. 
	The morphism $p$ is obtained as in Lemma~\ref{lem:existence-extension-morphism}; in particular, $k_1 \circ p = e_1^{\prime\prime} \circ u_1$. 
	Using this, we compute 
	\begin{align*}
		k_1 \circ p \circ b_{\kprime}^{\prime}	& = e_1^{\prime\prime} \circ u_1 \circ b_{\kprime}^{\prime} \\
																						& = e_1^{\prime\prime} \circ b_{l_1} \circ s_L \\
																						& = k_1 \circ b_{\kprime} \enspace .
	\end{align*}
	In particular, since $k_1$ is a monomorphism, $p \circ b_{\kprime}^{\prime} = b_{\kprime}$. 
	This makes the upper part of Fig.~\ref{fig:ensuring-correct-interface} commute. 
	Again, the universal property of the left pushout implies the existence of a morphism $\pprime$ that makes the whole diagram commute. 
	Furthermore, pushout decomposition implies the second square to be a pushout, as desired. 
	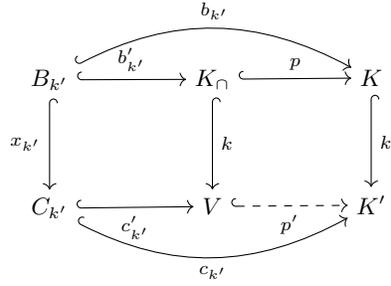
\begin{figure}
		\centering
		\begin{tikzpicture}
			\matrix (m) [	matrix of math nodes,
										nodes in empty cells,
										row sep=1.5em,
										column sep=1.5em,
										minimum width=1.5em]
			{
				B_{\kprime}	&	&	\Kcap		&	& K \\
										&	&					&	&	\\
				C_{\kprime} &	& V	&	&	\Kprime \\};
			\path[-stealth]
				(m-1-3) edge [right hook->] node [right] {\scriptsize $k$} (m-3-3)
								edge [right hook->] node [above] {\scriptsize $p$} (m-1-5)
				(m-1-1) edge [right hook->] node [above] {\scriptsize $b_{\kprime}^{\prime}$} (m-1-3)
								edge [right hook->] node [left] {\scriptsize $x_{\kprime}$} (m-3-1)
								edge [right hook->,bend left=28] node [above] {\scriptsize $b_{\kprime}$} (m-1-5)
				(m-1-5) edge [right hook->] node [right] {\scriptsize $\kprime$} (m-3-5)
				(m-3-1) edge [right hook->] node [below] {\scriptsize $c_{\kprime}^{\prime}$} (m-3-3)
								edge [right hook->,bend right=28] node [below] {\scriptsize $c_{\kprime}$} (m-3-5)
				(m-3-3) edge [right hook->,dashed] node [below] {\scriptsize $\pprime$} (m-3-5);
		\end{tikzpicture}
		\caption{Ensuring $k$ to compute $\Kprime$.}
		\label{fig:ensuring-correct-interface}
	\end{figure}
	Summarizing, we constructed a common kernel $k$ that is compatible with $E$ and computes the given object $\Kprime$ as interface. \qed
\end{proof}

\begin{proof}[of Corollary~\ref{cor:derivable-gcrs-graph}]
	The categories of graphs, typed graphs, and attributed graphs are all known to meet the conditions of Proposition~\ref{prop:enhancement-characterization-monic}, i.e., they are $\M$-adhesive categories and have initial pushouts. 
	Thus, the first part of the statement follows directly from the set-theoretic characterization of initial pushouts in these categories. 
	
	For the second statement, consider the case of graphs without edges. 
	Given a concurrent rule $\conrule$, there are 
	\begin{equation*}
		\sum_{i=0}^{\min(|L_1 \setminus K_1|,|R_2 \setminus K_2|)} i! \cdot \binom{|L_1 \setminus K_1|}{i} \cdot \binom{|R_2 \setminus K_2|}{i}
	\end{equation*}
	ways to derive a generalized concurrent rule from it as one extends $K$ to $\Kprime$ by adding $i$ elements to it which have to be mapped injectively to elements from $L_1 \setminus K_1$ and $R_2 \setminus K_2$, respectively. \qed
\end{proof}

\begin{proof}[of Proposition~\ref{prop:preservation-property}]
	Let $G_0 \xhookleftarrow{g_0} D \xhookrightarrow{g_2} G_2$ be a span stemming from an application of $\conrule$ at match $m$ and $d: K \to D$ the according morphism from the interface of the rule to the context object of this transformation. 
	Let $\kprime: K \hookrightarrow \Kprime$ be the enhancement morphism of $\genconrule$, i.e., the unique $\M$-morphism with $\lprime \circ \kprime = l$ and $\rprime \circ \kprime = r$ provided by construction. 
	Compare Fig.~\ref{fig:construct-transformation} for the following. 
	
	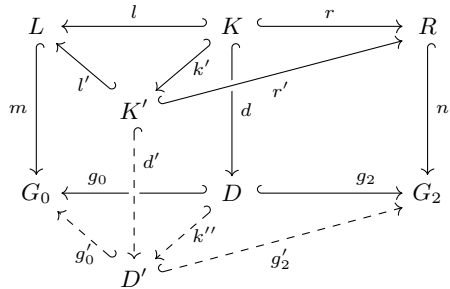
\begin{figure}
		\centering
		\begin{tikzpicture}
			\matrix (m) [	matrix of math nodes,
										nodes in empty cells,
										nodes={anchor=center},
										row sep=2em,
										column sep=2em,
										minimum width=2em]
			{
				L		&							&	K	&	& R	\\
						& \Kprime			&		&	&	\\
				G_0	&							& D	&	&	G_2 \\
						&	D^{\prime}	&		&	& \\};
			\path[-stealth]
				(m-1-1) edge [right hook->] node [left] {\scriptsize $m$} (m-3-1)
				(m-1-3) edge [right hook->] node [right] {\scriptsize $d$} (m-3-3)
								edge [left hook->] node [above] {\scriptsize $l$} (m-1-1)
								edge [right hook->] node [above] {\scriptsize $r$} (m-1-5)
								edge [left hook->] node [right] {\scriptsize $\kprime$} (m-2-2)
				(m-1-5) edge [right hook->] node [right] {\scriptsize $n$} (m-3-5)
				(m-3-3) edge [left hook->] node [above,near end] {\scriptsize $g_0$} (m-3-1)
								edge [right hook->] node [above,near end] {\scriptsize $g_2$} (m-3-5)
								edge [left hook->,dashed] node [right] {\scriptsize $\kprimeprime$} (m-4-2)
				(m-4-2) edge [left hook->,dashed] node [below] {\scriptsize $g_0^{\prime}$} (m-3-1)
				(m-4-2.10) edge [right hook->,dashed] node [below] {\scriptsize $g_2^{\prime}$} (m-3-5.210)
				(m-2-2) edge [left hook->] node [below] {\scriptsize $\lprime$} (m-1-1)
								edge [-,draw=white, line width=4pt] (m-4-2)
								edge [right hook->,dashed] node [right,near start] {\scriptsize $d^{\prime}$} (m-4-2)
				(m-2-2.10) edge [-,draw=white, line width=4pt] (m-1-5.210)
								edge [right hook->] node [below] {\scriptsize $\rprime$} (m-1-5.210);
		\end{tikzpicture}
		\caption{Constructing transformation via generalized concurrent rule from transformation via the concurrent rule.}
		\label{fig:construct-transformation}
	\end{figure}
	
	First, compute $K^{\prime} \xrightarrow{d^{\prime}} D^{\prime} \xhookleftarrow{\kprimeprime} D$ as pushout of $K^{\prime} \xhookleftarrow{\kprime} K \xrightarrow{d} D$. 
	Note that $\kprimeprime \in \M$ as it arises by pushout along $\kprime \in \M$. 
	In particular, $\kprimeprime$ is an isomorphism if and only if $\kprime$ is one (as this pushout is a pullback as well). 
	Moreover, one computes
	\begin{equation*}
		m \circ \lprime \circ \kprime = m \circ l = g_0 \circ d \text{ and } n \circ \rprime \circ \kprime = n \circ r = g_2 \circ d \enspace .
	\end{equation*}
	This means, by the universal property of that pushout, we obtain morphisms $g_i^{\prime}: D^{\prime} \hookrightarrow G_i$ such that $g_i^{\prime} \circ \kprimeprime = g_i$ for $i= 0,2$. 
	By pushout decomposition, both induced squares (the front squares in Fig.~\ref{fig:construct-transformation}) are pushouts. 
	Moreover, $m \models \mathit{ac}$, where $\mathit{ac}$ is the application condition of $\genconrule$, since $\mathit{ac}$ is also the application condition of $\conrule$ (compare Definitions~\ref{def:concurrent-rule} and~\ref{def:generalized-concurrent-rule}) and $\conrule$ is applicable at $m$.
	In particular, $G_0 \xhookleftarrow{g_0^{\prime}} D^{\prime} \xhookrightarrow{g_2^{\prime}} G_2$ is a transformation from $G_0$ to $G_2$ via $\genconrule$ at match $m$ and the desired morphism $\kprimeprime$ exists. \qed
\end{proof}

\begin{proof}[of Theorem~\ref{thm:generalized-concurrency-theorem}]
	The \emph{synthesis case} holds by virtue of the synthesis case of the Concurrency Theorem (see, e.g., \cite[Theorem~4.17]{EGHLO14} for its statement in the context of $\M$-adhesive categories and rules with application conditions) and Proposition~\ref{prop:preservation-property}: 
	Whenever such an $E$-related sequence of applications of $\rho_1$ and $\rho_2$ is given, the transformation $G_0 \Rightarrow_{\conrule,m} G_2$ exists by the Concurrency Theorem, and, hence, the transformation $G_0 \Rightarrow_{\genconrule,m} G_2$ by Proposition~\ref{prop:preservation-property}.

	For the 
	\emph{analysis case}, 
	it suffices to show that $\conrule$ is applicable at match $m$ if $\rho_1$ is applicable at $m \circ e_1^{\prime}$. 
	Applicability of $\rho_2$ at a suitable match then, again, follows by the analysis case of the Concurrency Theorem.
	The second direction of the stated equivalence is trivial. 
	Figure~\ref{fig:gen-con-thm_analysis_monic} displays the substance of that proof; the solid lines are given and the dashed ones are constructed throughout the proof. 
	
	\begin{figure}
		\centering
		\includegraphics[width=.75\textwidth]{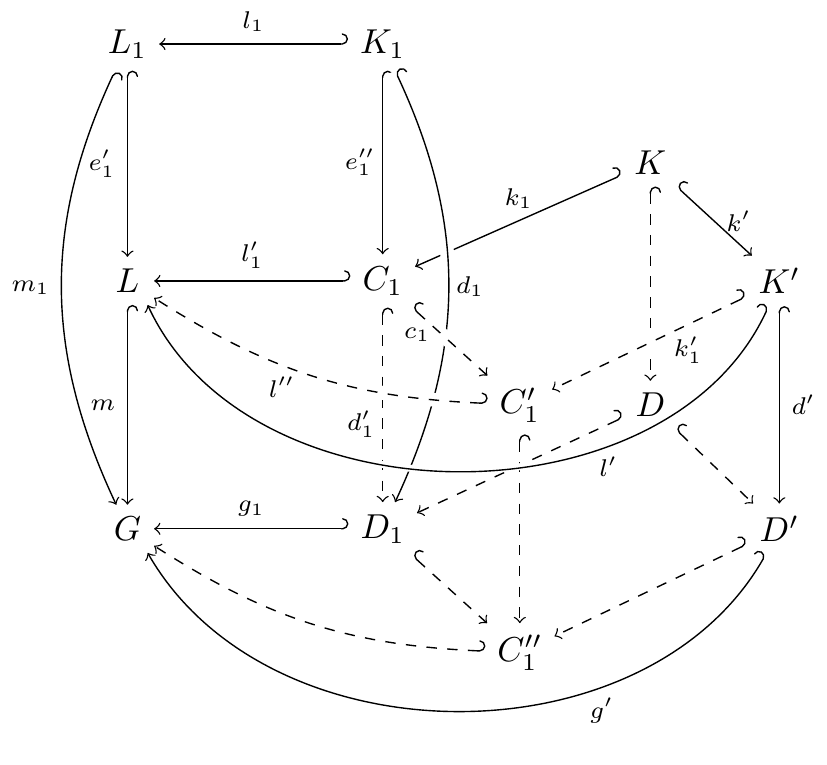}%
		\caption{Proving $\conrule$ to be applicable at $m$.}%
		\label{fig:gen-con-thm_analysis_monic}%
	\end{figure}
	
	First, the outer square to the left is the pushout (with context object $D_1$) that exists because of the applicability of $\rho_1$ at $m \coloneqq m \circ e_1^{\prime}$. 
	Furthermore, the solid bent square at the front is the pushout (with context object $D^{\prime}$) that exists because of the applicability of $\genconrule$ at $m$. 
	To show that also $\conrule$ is applicable at $m$, it suffices to construct a pushout complement $D$ for $m \circ l$ with $l \coloneqq l_1^{\prime} \circ k_1$. 
	We construct this pushout as composition of two pushouts.
	
	First, since $m,l_1 \in \M$, $\M$-pushout-pullback decomposition ensures that pulling back $m$ and $g_1$ decomposes the outer left square into two pushouts. 
	However, the unique pushout complement of $e_1^{\prime} \circ l_1$ is known to be given by $l_1^{\prime} \circ e_1^{\prime\prime}$ such that (up to isomorphism) $C_1$ is the object resulting from pulling back $m$ and $g_1$. 
	In particular, we obtain $d_1^{\prime}: C_1 \hookrightarrow D_1$ such that both the left squares are pushouts. 
	
	Subsequently, we compute $C_1^{\prime}$ as pushout of $k_1$ and $\kprime$. 
	Its universal property induces the morphism $\llprime$; in particular, $\lprime = \llprime \circ k_1^{\prime}$ and $\llprime \in \M$ by the existence of $\M$-effective unions. 
	Then, $C_1^{\prime\prime}$ is computed as pushout of $c_1$ and $d_1^{\prime}$; again, the morphism $C_1^{\prime\prime} \hookrightarrow G$ is obtained by its universal property. 
	Now, invoking $\lprime = \llprime \circ k_1^{\prime}$, pushout decomposition, and uniqueness of $C_1^{\prime\prime}$ as pushout complement for $m \circ \llprime$, the two vertical squares $C_1^{\prime} \xhookrightarrow{\llprime} L \xhookrightarrow{m} G \hookleftarrow C_1^{\prime\prime} \hookleftarrow C_1^{\prime}$ and $\Kprime \xhookrightarrow{k_1^{\prime}} C_1^{\prime} \hookrightarrow C_1^{\prime\prime} \hookleftarrow D^{\prime} \xhookleftarrow{d^{\prime}} K^{\prime}$ can be recognized to decompose the bent pushout at the front (that is given by the applicability of $\genconrule$ at $m$) into two pushouts. 
	
	In particular, we obtained the top and the front faces of the right cube and they are all pushouts. 
	Completing the cube by computing $D$ as pullback and invoking the vertical weak van Kampen property of the right front face, the left back face of that cube is a pushout, as well. 
	This means, since $l = l_1^{\prime} \circ k_1$, composing that pushout with the lower pushout of the two left ones constitutes the left pushout of a transformation of $\conrule$ at match $m$ with context object $D$. 
	In particular, $\conrule$ is applicable at $m$. 
	
	Finally, we did not explicitly mention application conditions but dealt with them implicitly, invoking the Concurrency Theorem for rules with application conditions. 
	To at least somewhat motivate the definition of the application condition of a GCR (or concurrent rule) and for the convenience of the reader, we reproduce the relevant computation that can also be found in the proof of \cite[Theorem~4.17]{EGHLO14}: 
	Whenever $m_1,m_2$ are the matches for $\rho_1,\rho_2$ constituting an $E$-related transformation sequence and $m$ is the corresponding match for the (generalized) concurrent rule (or vice versa in case analysis is possible), we have
	\begin{align*}
		m_1 \models \mathit{ac}_1 \text{ and } m_2 \models \mathit{ac}_2	& \iff m \models \Sh(e_1^{\prime},\mathit{ac}_1) \text{ and } h \models \Sh(h,\mathit{ac}_2) \\
																																			& \iff m \models \Sh(e_1^{\prime},\mathit{ac}_1) \text{ and } m \models \Le(\pprime,\Sh(h,\mathit{ac}_2)) \\
																																			& \iff m \models \Sh(e_1^{\prime},\mathit{ac}_1) \wedge \Le(\pprime,\Sh(h,\mathit{ac}_2)) \\
																																			& \iff m \models \mathit{ac}
	\end{align*}
	where $\pprime = L \xhookleftarrow{l_1^{\prime}} C_1 \xhookrightarrow{r_1^{\prime}} E$, $m_1 = m \circ e_1^{\prime}$, and $h: E \to G_1$ is the morphism with $h \circ e_1 = n_1$ and $h \circ e_2 = m_2$ that exists because the transformation sequence is $E$-related. 
	The computation relies on Facts~\ref{fact:correctness-shift-morphism} and \ref{fact:correctness-shift-rule}. \qed
\end{proof}

}{}

\end{document}